\documentclass[twoside,a4paper]{article}
\usepackage{amsmath,graphicx,amssymb,fancyhdr,amsthm}
\newtheorem{thm}{Theorem}[section]

\newtheorem{lem}[thm]{Lemma}
\newtheorem{prop}[thm]{Proposition}

\theoremstyle{definition}
\newtheorem{defn}[thm]{Definition}
\theoremstyle{remark}

\def\beq{\begin{eqnarray}}
\def\eeq{\end{eqnarray}}
\def\bsp{\begin{split}}
\def\esp{\end{split}}

\newcommand{\mbold}[1]{\mbox{\boldmath{\ensuremath{#1}}}}
\newcommand{\RR}{\mathbb{R}}

\def\bl {\bf{\ell}}

\def \bl {\mbox{\boldmath{$\ell$}}}

\newcommand{\el}{\mbox{\boldmath $\ell$}}
\newcommand{\en}{\mbox{\boldmath $n$}}
\newcommand{\emm}{\mbox{\boldmath $m$}}
\newcommand{\emmB}{\mbox{\boldmath $\bar m$}}

\newcommand{\be}{\begin{eqnarray}}
\newcommand{\ee}{\end{eqnarray}}
\newcommand{\bse}{\begin{subequations}}
\newcommand{\ese}{\end{subequations}}
\newcommand{\bdm}{\begin{displaymath}}
\newcommand{\edm}{\end{displaymath}}

\def \tS {\tilde{S}}

\def \tq {\tilde{q}}
\def \tr {\tilde{r}}
\def \ts {\tilde{s}}

\begin{document}
\title{\Large\textbf{Kundt Spacetimes}}
\author{{\large\textbf{A. Coley$^{1}$, S. Hervik$^{2,~ 1}$, G. Papadopoulos$^{1}$, N. Pelavas$^{1}$}}
\vspace{0.3cm} \\
$^{1}$Department of Mathematics and Statistics,\\
Dalhousie University,
Halifax, Nova Scotia,\\
Canada B3H 3J5
\vspace{0.2cm}\\
$^{2}$Faculty of Science and Technology,\\
 University of Stavanger,\\  N-4036 Stavanger, Norway
\vspace{0.3cm} \\
\texttt{aac, gopapad, pelavas@mathstat.dal.ca, sigbjorn.hervik@uis.no} }
\date{\today}
\maketitle
\pagestyle{fancy}
\fancyhead{} 
\fancyhead[EC]{A. Coley, S. Hervik, G.O. Papadopoulos, N. Pelavas}
\fancyhead[EL,OR]{\thepage}
\fancyhead[OC]{Kundt Spacetimes}
\fancyfoot{} 
\begin{abstract}

Kundt spacetimes are of great importance in general relativity in
four dimensions and have a number of physical applications in
higher dimensions in the context of string theory. The degenerate
Kundt spacetimes have many special and unique mathematical
properties, including their invariant curvature structure and
their holonomy structure. We provide a rigorous geometrical
kinematical definition of the general Kundt spacetime in four
dimensions; essentially a Kundt spacetime is defined as one
admitting a null vector that is geodesic, expansion-free,
shear-free and twist-free. A Kundt spacetime is said to be
degenerate if the preferred kinematic and curvature null frames
are all  aligned. The degenerate Kundt spacetimes are the only
spacetimes in four dimensions that are not
$\mathcal{I}$-non-degenerate, so that they are not determined by
their scalar polynomial curvature invariants. We first discuss the
non-aligned Kundt spacetimes, and then turn our attention to the
degenerate Kundt spacetimes. The degenerate Kundt spacetimes are
classified algebraically by the Riemann tensor and its covariant
derivatives in the aligned kinematic frame; as an example, we
classify Riemann type D degenerate Kundt spacetimes in which
$\nabla(Riem),\nabla^{(2)}(Riem)$ are also of type D. We discuss
other local characteristics of the degenerate Kundt spacetimes.
Finally, we discuss degenerate Kundt spacetimes in higher
dimensions.

\end{abstract}

\newpage


\section{Introduction}

Kundt spacetimes have very interesting mathematical properties and
are of great importance in general relativity in 4 dimensions and
have a number of topical applications in higher dimensions in the
context of string theory. {\footnote{The primary focus of this
paper is on the mathematical properties of Kundt spacetimes. The
physical motivation for studying these spacetimes, particularly in
higher dimensions, has been discussed previously (however, see the
discussion and references in section 7.2).}} The Kundt spacetimes,
and especially the degenerate Kundt subclass (to be defined
below), have many special and unique mathematical properties. The
degenerate Kundt metrics are the only metrics not determined by
their scalar curvature invariants \cite{inv}, and they have
extraordinary holonomy structure \cite{johan}. In particular, it
is not possible to define a unique timelike curvature operator in
a degenerate Kundt spacetime and consequently there is no unique,
preferred timelike direction associated with its geometry (through
the curvature and its covariant derivatives). Therefore, there is
no intrinsic $1 + 3$ split of the degenerate Kundt spacetime in
general relativity (or a $1 + (n-1)$ split in general).

In particular, for any given Lorentzian spacetime
$(\mathcal{M},g)$ we denote by, $\mathcal{I}$, the set of all
scalar polynomial curvature invariants constructed from the Riemann tensor
and its covariant derivatives. If there does not exist a
metric deformation of $g$ having the same set of invariants as
$g$, then we call the set of invariants $\mathcal{I}$-\emph{non-degenerate}, and the spacetime metric $g$ is called
\emph{$\mathcal{I}$-non-degenerate}. This means that for a metric
which is $\mathcal{I}$-non-degenerate, the invariants locally
characterize the spacetime uniquely. In \cite{inv} it was proven
that a four-dimensional (4D) Lorentzian spacetime metric is either
{\em $\mathcal{I}$-non-degenerate or degenerate Kundt}. This
striking result implies that metrics not determined by their
scalar polynomial curvature invariants  (at least locally) must be
of degenerate Kundt form.

We first present a rigorous geometrical kinematical definition of
the general Kundt spacetime in 4 dimensions, and discuss some of
its (local) invariant characteristics. We present some necessary
conditions, in terms of the curvature components, for a space to
define (locally) a Kundt geometry and discuss the issue of
complete gauge fixing. Although section 2.1 is essentially review,
a number of new results are presented in sections 2.2 and 2.3.

The 4D Kundt class is defined as those spacetimes admitting a null
vector $\ell$ that is geodesic, expansion-free, shear-free and
twist-free \cite{Higher,kramer}; it follows that in 4D there
exists a {\em kinematic} frame with
$\kappa=\sigma=\rho=\epsilon=0$. We first discuss the non-aligned
(non-degenerate) Kundt spacetimes (and show that this set is
non-trivial). We then turn to the degenerate Kundt spacetimes. All
of the results presented in sections 3--6 are new.

In a  degenerate Kundt spacetime, the kinematical frame and the
Riemann $\nabla^k(Riem)$ type II aligned null frame are all {\it
aligned}.  We algebraically classify the degenerate Kundt
spacetimes in terms of their Riemann type in the aligned kinematic
frame, or more finely by their Ricci and Weyl types separately.
Within each algebraic type, it is also useful to classify the
covariant derivative(s) of the Riemann tensor (and particularly
$\nabla(Riem)$ and $\nabla^2(Riem)$) in terms of their algebraic
types. As an example, we classify
$\nabla(Riem),\nabla^{(2)}(Riem)$ type D degenerate Kundt
spacetimes. {\footnote{This is the first time the covariant
derivative(s) of the Riemann tensor have been explicitly
classified algebraically.}}

We then discuss the Kundt spacetimes and their invariant
properties. Degenerate Kundt spacetimes in 4D are not
\emph{$\mathcal{I}$}-non-degenerate, so that  the scalar
polynomial curvature invariants do not locally characterize the
spacetime uniquely \cite{inv}. We discuss the equivalence problem
for degenerate Kundt spacetimes, focussing on the special type D
subclass, and discuss Kundt spacetimes and scalar invariants.

Finally, we turn to higher dimensions. The (general) higher
dimensional Kundt class is defined in a very similar manner to the
4D class \cite{Higher}.  We review important examples of
degenerate Kundt spacetimes in higher dimensions and discuss
applications in the context of string theory in sections 7.1 and
7.2. We note that many of the results discussed in this paper can
be generalized to higher dimensions; in particular, two important
higher dimensional theorems are presented in section 7.3.

Let us remark on the technical assumptions made in this paper. The
following theorems hold on neighbourhoods where the Riemann, Weyl
and Segre types do not change.  In the algebraically special cases
we also need to assume that the algebraic type of the
higher-derivative curvature tensors also do not change, up to the
appropriate order.  Most crucial is the definition of the
curvature operators and in order for these to be well defined the
algebraic properties of the curvature tensors need to remain the
same over a neighbourhood. Finally, we note that extensive use is
made of the Newman-Penrose (NP) formalism \cite{NP,kramer} and
many calculations are done using GRTensor II \cite{grtensor}.



\section{Kundt geometries in 4 dimensions}
It is the purpose of this section to provide a complete local
account of the Kundt geometries in all their generality. The first
subsection contains a kinematical definition which implies a
proposition which, in turn, implies a lemma. All three can be used
to invariantly describe/characterise (locally) a Kundt geometry in
kinematical terms. The second subsection gives some necessary
conditions, in terms of the curvature components, for a space to
define (locally) a Kundt geometry. The third subsection attacks
the same problem,  but from the point of view of complete gauge
fixing.

\subsection{Definition, Proposition and Lemma}
The \textit{Kundt} geometry is defined (cf.\ \cite{Kundt} and \cite{kramer}) as:
\begin{defn}\label{definition}
A space $\mathcal{S}$ defines (locally) a Kundt geometry in 4
dimensions if there exists (locally) a null congruence of which
the tangent vector field is hypersurface orthogonal
\footnote{Other terms are also in use such as:
\textit{non-rotating}, \textit{normal}, \textit{non-twisting}
--although the latter is reserved for the null congruences only.}
(and therefore geodesic), non-diverging \footnote{Or
\textit{non-expanding}.}, and non-shearing.
\end{defn}
\noindent In the original paper \cite{Kundt}, and within the
context of the NP formalism, appeal to the energy conditions is
made: a space $\mathcal{S}$ defines (locally) a Kundt geometry in
4 dimensions if there exists (locally) a null congruence
$\el(\equiv l^{a}\partial_{a})$,
of which the tangent vector field is hypersurface orthogonal (and therefore geodesic),
non-diverging, and obeying the energy condition $T_{ab}l^{a}l^{b}\geqslant 0$; the latter,
through the Ricci identities and the Einstein's equations, implies that it will be non-shearing as well
\footnote{A null congruence  of which the tangent vector field is hypersurface orthogonal
(and therefore geodesic), and non-shearing defines a \textit{ray congruence} --see \cite{Kundt} and the references therein).}.\\
By virtue of the Ricci identities one could replace the energy
condition by the vanishing of the shear, thus providing a
completely kinematic definition (i.e., that given at the
beginning). In fact, this attitude is partially adopted in
\cite{kramer} (Chapter 31, section 31.2; but immediately after the
authors consider further restrictions, like the conditions of the
\textit{Goldberg-Sachs} theorem).

Based on the definition \eqref{definition} one can prove the following:
\begin{prop}\label{proposition}
A space $\mathcal{S}$ defines (locally) a Kundt geometry in 4 dimensions if and only if there exists (locally) a family of frames $\{\mathbf{e}_{A}\}=\{\el,\en,\emm,\emmB\}$ (i.e., defined up to null rotations about $\el$) characterized by the conditions that:
$$\epsilon=\kappa=\rho=\sigma=0$$
for the spin connexion coefficients. Also, the generic line
element is given, in a local coordinate system
$\{x^{a}\}=\{u,v,\zeta,\overline{\zeta}\}$ (coordinates $u$ and
$v$ are real, while coordinate $\zeta$ is complex): \be
ds^{2}=2du(Hdu+dv+Wd\zeta+\overline{W}d\overline{\zeta})-2P^{-2}d\zeta
d\overline{\zeta}, ~~~\partial_{v}P=0 \ee
\end{prop}
\begin{proof}
Let $\mathcal{S}$ be a space  which defines (locally) a Kundt geometry. The optimal way towards the proof of the assertion would perhaps be provided by the observation that the NP equations are invariant under arbitrary frame transformations (for the case under consideration, \textit{Lorentz transformations}). Therefore, this freedom should be implemented in order to gauge fix (even partially) the frame relative to the geometric structure(s) assumed for the space.\\
Indeed, let $\el'$ denoting the vector field which is tangent to the congruence under consideration. The effect of a boost transformation \cite{Car}:
\begin{align*}
\el' & \rightarrow \el, ~~~\el=z\overline{z}^{-1}\el'
\end{align*}
with a proper value for the complex parameter $z$, can make the
vector field $\el$ to be not only a geodesic (as a consequence of
being both null and normal) but also affinely parameterized; i.e.,
$\kappa=0$, and $\epsilon+\overline{\epsilon}=0$ , respectively.
In this gauge, the properties for being normal, non-diverging and
non-shearing correspond to the vanishing of
$\rho-\overline{\rho}$, $\rho+\overline{\rho}$, and $\sigma$,
respectively (see mainly \cite{Sa}, but also \cite{NP,Ch} as
well). Therefore: \bse\label{kundt conditions}
\begin{align}
\epsilon+\overline{\epsilon}=0\\
\kappa=0\\
\rho=0\\
\sigma=0
\end{align}
\ese Since $\el$ is normal, the corresponding covector (also
denoted by $\el$) complies with all the conditions of the
\textit{Frobenius' Theorem}. The latter guarantees (at least
locally) that there exist two non-constant, non-zero, real
functions $F$, and $u'$ such that $\el=Fdu'$. Another
non-constant, non-zero, real function $u$ can be defined such that
$du=Fdu'$. In addition, if $v$ denotes an affine parameter, then a
local coordinate system $\{y^{a}\}=\{u,v,\chi^{1},\chi^{2}\}$ can
be adopted, where all coordinates are real, such that
$\el=\partial_{v}$ or, alternatively, $l^{a}\equiv
dy^{a}/dv=(0,1,0,0)$. Then: \bse\label{first frame}
\begin{align}
\el &=l^{a}\partial_{a} & l^{a}&=(0,1,0,0)\\
\el &=l_{a}dy^{a}=du & l_{a}&=(1,0,0,0)\\
\en &=n_{a}dy^{a}=N_{0}du+dv+N_{A}d\chi^{A} & n_{a}&=(N_{0},1,N_{1},N_{2})\\
\emm &=m_{a}dy^{a}=M_{0}du+M_{A}d\chi^{A} & m_{a}&=(M_{0},0,M_{1},M_{2})
\end{align}
\ese
(where $A$ $\in$ $\{1,2\}$) by completing the frame.\\
In this coordinate system, the line element assumes the form:
\begin{align}\label{first line element}
ds^{2}&=2du\big((N_{0}-M_{0}\overline{M_{0}})du+dv+(N_{A}-M_{0}\overline{M_{A}}-\overline{M_{0}}M_{A})d\chi^{A}\big)\notag\\
&-2M_{A}\overline{M_{B}}d\chi^{A}d\chi^{B}
\end{align}
But, $\sigma=-i\star(\el\wedge\emm\wedge d\emm)$, and
$\rho=\frac{i}{2}\star\big(\el\wedge(\emmB\wedge d\emm-\emm\wedge
d\emmB-\en\wedge d\el)\big)$ (the $\star$ denotes the
\textit{Hodge dual}) as the first Cartan Structure Equations
imply. The implication of the vanishing of these two spin
connexion coefficients results in $\partial_{v}M_{A}=0$. A general
coordinate transformation, which as such leaves all the NP
quantities invariant, of the form:
\begin{align*}
\{u,v,\chi^{A}\} & \rightarrow \{u,v,\psi^{A}\}\equiv\{u,v,\zeta,\overline{\zeta}\}, ~~~\zeta=\frac{1}{\sqrt{2}}\big(f^{1}(u,\chi^{A})+if^{2}(u,\chi^{A})\big)
\end{align*}
where the $f^{A}$ are real functions, renders the $2-$dimensional
part of the previous line element,
$M_{A}\overline{M_{B}}d\chi^{A}d\chi^{B}$ apparently conformally
flat\footnote{because, any $2-$dimensional metric is conformally
flat.}; say
$P^{-2}d\zeta d\overline{\zeta}$, with $\partial_{v}P=0$.\\
Hence, in the local coordinate system
$\{x^{a}\}=\{u,v,\zeta,\overline{\zeta}\}$, the line element
\eqref{first line element} assumes the form: \be
ds^{2}=2du(Hdu+dv+Wd\zeta+\overline{W}d\overline{\zeta})-2P^{-2}d\zeta
d\overline{\zeta}, ~~~\partial_{v}P=0 \ee
where the coordinates $u$ and $v$ and the functions $H$ and $P$ are real, while the coordinate $\zeta$ and the function $W$ are complex.\\
Finally, the effect of a spin transformation \cite{Car}, which
preserves \eqref{kundt conditions}, can be used to set the
imaginary $\epsilon$ equal to zero. Then, the conditions: \bse
\begin{align}
\epsilon=0\\
\kappa=0\\
\rho=0\\
\sigma=0
\end{align}
\ese
are invariant under null rotations about the $\el$ vector field, defining thus a family of
frames. This completes the first logical direction of the proposition. The other direction is
straightforward to prove.
\end{proof}
Based on proposition \eqref{proposition}, it is possible to prove
the following:
\begin{lem}\label{lemma}
A space $\mathcal{S}$ defines (locally) a Kundt geometry in 4 dimensions if and only if there exists (locally) a family of frames $\{\mathbf{e}_{A}\}=\{\el,\en,\emm,\emmB\}$ (i.e., defined up to null rotations about $\el$ and spin-boosts) characterized by the conditions that:
$$\kappa=\rho=\sigma=0$$
for the spin connection coefficients.
\end{lem}
\begin{proof}
Indeed, a spin-boost transformation, which preserves the vanishing
of the coefficients $\kappa$, $\rho$, $\sigma$, can be used to
covert a non vanishing $\epsilon$ to a vanishing quantity and vice
versa.
\end{proof}

\subsection{Necessary Conditions}
Based on the Proposition and the Lemma, in principle we can deduce
some necessary conditions which must be met in order for a general
space time, given in an arbitrary frame, to (locally) define a
Kundt geometry.

Indeed, let $\mathcal{S}$ be a general space defined locally in
terms of the completely arbitrary local frame
$\{\mathbf{e}_{A}\}=\{\el,\en,\emm,\emmB\}$. Then, according to
the proofs of Proposition 2.1 and Lemma 2.1, it is easily inferred
that $\mathcal{S}$ defines (locally) a Kundt geometry if and only
if a null rotation about the $\en$ (co)vector (with a proper
complex parameter $z$) suffices to make the three connection
coefficients $\kappa'$, $\rho'$ and $\sigma'$ in the primed frame
vanish. However, in that frame these conditions, by virtue of the
NP equations, correspond to the following curvature conditions
(cf. the next subsection): \bse\label{necessary_conditions}
\begin{align}
\Psi'_{0}&=\Phi'_{00}=0\\
\Psi'_{1}&=\Phi'_{01}.
\end{align}
\ese An observation is pertinent at this point. These are not all
the curvature conditions implied by the vanishing of the spin
connection coefficients; in general, we need to completely gauge
fix the frame in order for all the curvature conditions to emerge.
Then, we would have to establish the one-to-one correspondence
amongst the spin connection coefficient conditions and the ensuing
curvature conditions. The equations above comprise just a subset
of the totality of the curvature conditions. Moreover, the
previously mentioned one-to-one correspondence has not been
established. From this point of view, these conditions are simply
necessary but not sufficient.

Using the well known transformation laws of the NP variables under
the action of the Lorentz group, and specifically here the null
rotations about the (co)vector $\en$, conditions
\eqref{necessary_conditions} assume the following form (in the
initial, arbitrary frame): \bse
\begin{align}
&\Psi_{0}+4z\Psi_{1}+6z^{2}\Psi_{2}+4z^{3}\Psi_{3}+z^{4}\Psi_{4}=0\\
&\Phi_{00}+2\overline{z}\Phi_{01}+2z\Phi_{10}+4z\overline{z}\Phi_{11}
+\overline{z}^{2}\Phi_{02}+z^{2}\Phi_{20}\notag\\
&+2\overline{z}z^{2}\Phi_{21}+2\overline{z}^{2}z\Phi_{12}+z^{2}\overline{z}^{2}\Phi_{22}=0\\
&\Psi_{1}+3z\Psi_{2}+3z^{2}\Psi_{3}+z^{3}\Psi_{4}=\Phi_{01}+2z\Phi_{11}+\overline{z}\Phi_{02}\notag\\
&+2z\overline{z}\Phi_{12}+z^{2}\Phi_{21}+z^{2}\overline{z}\Phi_{22}.
\end{align}
\ese The objective here is to eliminate the Lorentz parameter $z$
from this system, but the complexity of the general system renders
this goal unrealistic.  However, an algorithmic procedure could be
established:  in the initial, arbitrary frame, one tries to solve
algebraically one of the first two equations (the choice has to do
with the ``extremal'' cases of Conformally or Ricci flat spaces).
As the Fundamental Theorem of Algebra states, any of these two is
susceptible to 4 solutions --say $\{z_{1},z_{2},z_{3},z_{4}\}$;
the degeneracy of this set reflects the Petrov or the
Petrov-Pleba\'{n}ski type, respectively.  Suppose that one chooses
to solve, for example, the first equation.  If the result is a set
of 4 independent solutions (i.e., Petrov type I), then the second
and the third equations must be satisfied separately for each of
these four solutions. Therefore, in that case eight independent
(in general) syzygies amongst the components of the Weyl and Ricci
tensors will emerge.  If the result is 3 independent solutions
(i.e., Petrov type II), then there emerges six independent
syzygies (in general), etc.

Of course, the entire analysis above comprises one of the two
possibilities, which corresponds to the alignment of the generic
frame in such a way that the $\el'$ vector is aligned with the
characteristic Kundt vector. The other possibility concerns the
alignment of the generic frame in such a way that the $\en'$
vector is aligned with the characteristic Kundt vector. In order
to get the set of necessary conditions for the latter case, one
has simply to use the $\circ$ operation which interchanges the
r\^{o}le of the $\el$ and $\en$ (co)vectors:
\begin{align}
\circ : \el \longleftrightarrow \en
\end{align}
along with changes like $\Psi_{0} \longleftrightarrow \Psi_{4}$ etc. (see, the standard literature on the NP formalism).

\subsection{Complete gauge fixing and necessary Curvature conditions}

It is expected that the special features which define, in
kinematical terms, a Kundt geometry will impose some constraints
upon the curvature components; i.e., not all the components of the
Riemann tensor (and the components of its covariant derivatives)
will be independent. A standard approach towards these conditions
is to implement the entire Lorentz freedom in order to completely
gauge fix the frame relative to the geometric structure(s) assumed
for the space.  Of course, there is no ``recipe'' for finding the
optimal gauge fixing; i.e., that fixing which will exhibit as many
characteristics as possible.  However, it is possible to postulate
some criteria. For instance, a gauge which renders many of the NP
equations purely algebraic is better (more useful) than another
one which will give fewer algebraic equations.  Thus, a suggested
way is the following.

Let
$\{\mathbf{e}''''_{A}\}=\{\el'''',\en'''',\emm'''',\emmB''''\}$ be
a family of local frames, such that the characteristic vector
field (in the definition of the Kundt geometries) is a linear
combination of the frame vectors. First, a null rotation about
$\en''''$ can be used to align the vector $\el'''$ with the
characteristic vector field. Then, the covector $\el'''$ satisfies
the condition $\el'''\wedge\ d\el'''=0 \Leftrightarrow \el'''=Fdu$
locally, for some non-constant, non-zero, real functions $F$, $u$
(Poincar\'{e}'s Lemma). Second, a boost in the $\el'''-\en'''$
plane can be used to affinely parameterize the vector field
$\el''$: $\el''=du$. This in turn implies: \bse
\begin{align}
\kappa''&=0\\
\epsilon''+\overline{\epsilon''}&=0\\
\rho''-\overline{\rho''}&=0\\
\tau''-\overline{\alpha''}-\beta''&=0
\end{align}
\ese by virtue of $d\el''=0$. The properties of being geodesic,
affinely parametrized, normal, non-diverging, and non-shearing
correspond to the conditions: \bse
\begin{align}
\kappa''&=0\\
\epsilon''+\overline{\epsilon''}&=0\\
\rho''&=0\\
\sigma''&=0
\end{align}
\ese
Third, a spin in the $\emm''-\emmB''$ plane, which preserves all the previous restrictions, can be used to set:
\be
\epsilon'=0
\ee
Finally, a null rotation about $\el'$, preserving again all the previous restrictions, is implemented to set:
\be
\gamma=0
\ee
Now, the completely gauged system reads:
\bse\label{gauge_fixing}
\begin{align}
\epsilon&=0\\
\gamma&=0\\
\kappa&=0\\
\rho&=0\\
\sigma&=0\\
\alpha&=\overline{\tau}-\overline{\beta}
\end{align}
\ese
Substitution of \eqref{gauge_fixing} into the system of the NP and Eliminant equations (cf. \cite{Ch}) plus some algebraic manipulation result in the following curvature conditions of zero order:
\bse
\begin{align}
\Psi_{0}&=0\\
\Phi_{00}&=0\\
\Phi_{01}-\Psi_{1}&=0\\
\Lambda-\beta\pi-\tau\pi+\tau\overline{\beta}-\beta\overline{\tau}-\tau\overline{\tau}
+\overline{\beta}\overline{\pi}-\overline{\tau}\overline{\pi}-\Phi_{11}-\Psi_{2}&=0\\
2\Lambda-2\beta\tau-2\tau\overline{\beta}+2\tau\overline{\tau}-\Phi_{02}+\Psi_{2}&=0\\
-2\Lambda-2\beta\pi+2\tau\overline{\beta}-2\tau\overline{\tau}+2\overline{\beta}\overline{\tau}
+2\overline{\beta}\overline{\pi}+\Phi_{02}-\Psi_{2}&=0
\end{align}
\ese The presence of the spin connection coefficients $\beta$,
$\tau$, $\pi$, should not concern us since in a completely gauged
system everything is invariant. So, all $+2$ boost weight terms
vanish, there is only one independent $+1$ boost weight term,
while not all 0 boost weight terms are independent. Unfortunately,
it seems that an elimination of the spin connection coefficients,
which would lead to gauge independent syzygies, is not possible.


\section{Kundt spacetimes}
The 4D Kundt spacetimes are those spacetimes admitting a null
vector $\ell$ that is geodesic, expansion-free, shear-free and
twist-free \cite{Higher,kramer}. Since $\ell$ is geodesic, without
loss of generality we can always chose a complete NP tetrad with
$\ell$ (so that the NP spin coefficient $\kappa=0$) in which
$\epsilon=0$ and hence the geodesic $\ell$ is affinely
parametrized. Demanding that $\ell$ is expansion-free, shear-free
and twist-free in this frame, then implies that $\sigma=\rho=0$.
Hence in the canonical Kundt frame, or {\em kinematic} frame,
$\kappa=\sigma=\rho=\epsilon=0$. {\footnote {For $\ell$ to be
affinely parametrized it is only necessary for
$\epsilon+{\overline\epsilon}=0$. However, since all of the
calculations that follow are done in the frame (gauge) in which
$\epsilon=0$ for convenience (it is always possible to choose such
a gauge without loss of generality), and this choice of gauge does
not affect the definitions of algebraic types (of, for example,
the Riemann tensor), for simplicity of presentation we have {\em
defined} the kinematic frame so that $\epsilon=0$ in addition to
$\kappa=\sigma=\rho=0$.}} It follows that in 4D there exists a
{\em kinematic} frame with
\begin{equation}
\kappa=\sigma=\rho=\epsilon=0, \label{kinframe}
\end{equation}
in which the (general) Kundt metric can always be written
\begin{equation}
d s^2=2d u\left[d v+H(v,u,x^k)d u+W_{i}(v,u,x^k)d x^i\right]+
g_{ij}(u,x^k)d x^id x^j. \label{Kundt}
\end{equation}
The~metric functions $H$, $W_{i}$ and $g_{ij}$  ($i,j = 2,3$) are
real and will satisfy additional conditions if the Einstein field
equations are applied. {\footnote{It is this real form for the
Kundt metric which is more appropriate for generalizing to higher
dimensions.}}

We note that in the above kinematic frame, the positive  boost
weight +2 terms of the Riemann tensor are automatically zero.
However, in the above kinematic frame, if $W_{,vv} \neq 0$, then
the Ricci tensor (and Riemann tensor) have positive boost weight
terms, and if $H_{,vvv} \neq 0$ the covariant derivatives of the
Riemann tensor have positive  boost weight terms.

\subsubsection{Higher dimensions}

The (general) higher dimensional Kundt class is defined in a very
similar manner (and is consistent with the definition of
generalized Kundt spacetimes in $n$-dimensions given, for example,
in \cite{Higher}); namely, the spacetimes admit a null vector
$\ell$ which satisfies $\ell_{a;b}\ell^b \propto \ell_a$ (which is
zero for an affine parameterization), $\ell^a_{~;a}=0$,
$\ell_{(a;b)}\ell^{(a;b)}=0$ and $\ell_{[a;b]}\ell^{[a;b]}=0$, so
that there exists a kinematic frame in which $\ell$ is geodesic,
expansion-free, shear-free and twist-free. In particular, the
$n$-dimensional Kundt metric can be written as in (\ref{Kundt}),
where now  $i,j = 2, ..., n-1$. In the kinematic Kundt null frame
\begin{equation}
{\bl}= d u, \quad {\bf n}= d v+Hd u+W_{ i}{{\bf m}}^{ i}, \label{null frame}
\quad {\bf m}^{ i}= m^{i}_{ j}d x^{j},
\end{equation}
such that $g_{ij}=m^{k}_{\ i}m_{kj}$ and where $m^{i}_{\ j}$ can
be chosen to be in upper triangular form by an appropriate choice
of frame. The metric (\ref{Kundt}) possesses a null vector field
${\bl} = \frac{\partial}{\partial{v}}$ which is  geodesic,
non-expanding, shear-free and non-twisting; i.e., $\ell_{i;j}=0$
\cite{Higher}. Since ${\bl}$ is affinely parametrized, the Ricci
rotation coefficients in the kinematic frame are thus given by:
\begin{equation}
\ell_{(a;b)} = L \ell_{a}\ell_{b} + L_i(\ell_{a} m^{i}_{b} +
\ell_{b} m^{i}_{a}).  \label{Rcoefs}\end{equation}

\subsection{Non-aligned Kundt metrics}
Consider (\ref{Kundt}) with $g_{ij}(u,x^k)=\delta_{ij}$ and define the standard Newman-Penrose
(NP) tetrad for the Kundt metric having an ${\bl}$ with $\kappa=\sigma=\rho=\epsilon=0$ (i.e., the kinematic frame).  It immediately follows from the NP equations that in the kinematic frame
\begin{eqnarray}
\Psi_{0}=\Phi_{00}=0, & \Psi_{1}=\Phi_{01}.   \label{npkinem}
\end{eqnarray}
Moreover, the two complex Weyl invariants $I$ and $J$ are nonzero,
and $I^3 - 27J^2 \neq 0$.  Therefore (\ref{Kundt}) is Petrov type
I (and hence from (\ref{npkinem}) also Riemann type I).

We shall now consider the possibility that the Weyl tensor is
Petrov type II with a Weyl aligned null direction (WAND) ${\bf k}$
not equal to ${\bl}$.  Clearly, if $\Psi_{1}=0$ then ${\bf
k}={\bl}$ and the Petrov type II Weyl aligned frame coincides with
the kinematic frame; henceforth we assume $\Psi_{1}\neq 0$.  The
following Lorentz transformations are applied successively in
order to normalize the Weyl tensor and determine conditions so
that ${\bf k}$ is a Petrov type II WAND.
\begin{enumerate}
\item[1)] Null rotations about ${\bl}$ ($c$ ; $\Psi^{(1)}_{A}$)
\item[2)] Spin-Boost ($a$, $\theta$ ; $\Psi^{(2)}_{A}$)
\item[3)] Null rotations about ${\bf n}$ ($b$ ; $\Psi^{(3)}_{A}$)
\end{enumerate}
\noindent (where in $(\cdot)$ we give the corresponding Lorentz
parameters and denote the Weyl scalars $\Psi^{(i)}_{A}$,
$A=0\ldots 4$, resulting from transformation $i)$.  By applying 1)
$c$ can always be chosen such that
\begin{equation}
\Psi^{(1)}_{4}=\Psi_{4}+4c\Psi_{3}+6c^2\Psi_{2}+4c^3\Psi_{1}=0.  \label{psi41}
\end{equation}
\noindent Next, apply 2) so that $\Psi^{(2)}_{1}=1$.  Since the
kinematic frame is invariant under 1) and 2) we still have
${\bl}^{(2)}$ with $\kappa^{(2)}=\sigma^{(2)}=\rho^{(2)}=0$.
Lastly, apply 3) and requiring that ${\bl}^{(3)}={\bf k}$ is a
Petrov type II WAND gives
\begin{eqnarray}
\Psi^{(3)}_{0} & = & 2b\left[2+3b\Psi^{(2)}_{2}+2b^2\Psi^{(2)}_{3}\right]  =  0 \label{psi03} \\
\Psi^{(3)}_{1} & = & 1+3b\Psi^{(2)}_{2}+3b^2\Psi^{(2)}_{3}  =  0  \label {psi13} \\
\Psi^{(3)}_{2} & = & \Psi^{(2)}_{2}+2b\Psi^{(2)}_{3}  \\
\Psi^{(3)}_{3} & = & \Psi^{(2)}_{3}  \\
\Psi^{(3)}_{4} & = & 0.
\end{eqnarray}
\noindent Notice, $\Psi^{(2)}_{2}\neq 0$ otherwise equations
(\ref{psi03}) and (\ref{psi13}) give a contradiction.  Solving
$\Psi^{(3)}_{0}=\Psi^{(3)}_{1}=0$ gives that the null rotation
parameter, $b$, is the root of $3b\Psi^{(2)}_{2}+4=0$ subject to
the constraint $9\Psi^{(2)^{\, 2}}_{2}=16 \Psi^{(2)}_{3}$.
Therefore, if there exists a Petrov type II WAND ${\bf
k}={\bl}^{(3)}\neq {\bl}$ then the nonzero Weyl scalars are
\begin{eqnarray}
\Psi^{(3)}_{2}=-\frac{1}{2}\Psi^{(2)}_{2}, & & \Psi^{(3)}_{3}=\Psi^{(2)}_{3},
\end{eqnarray}
\noindent and satisfy the following condition in the Weyl aligned
frame:
\begin{equation}
9\Psi^{(3)^{\, 2}}_{2}=4 \Psi^{(3)}_{3}.  \label{ptii3}
\end{equation}
\noindent Transforming (\ref{ptii3}) to the kinematic frame gives
\begin{equation}
16\Psi_{1}\Psi_{3}-9\Psi_{2}^{2}+12c\Psi_{1}\Psi_{2}+12c^2\Psi_{1}^{2}=0.  \label{ptiikin}
\end{equation}
\noindent Since $c$ is a root of (\ref{psi41}), we obtain the necessary condition
\begin{equation}
9\Psi_{2}^{3}-16\Psi_{1}\Psi_{2}\Psi_{3}+6\Psi_{1}^{2}\Psi_{4}+2c\Psi_{1}\left[3\Psi_{2}^2-4\Psi_{1}\Psi_{3}\right]=0, \label{ptiineccon}
\end{equation}
\noindent which implies the following two cases.

\noindent \emph{Case 1}: $3\Psi_{2}^2-4\Psi_{1}\Psi_{3}=0$.  From (\ref{ptiikin}) it is easily shown that
\begin{equation}
\Psi^{(1)}_{3}=\Psi_{3}+3c\Psi_{2}+3c^2\Psi_{1}=0.  \label{psi31}
\end{equation}
\noindent In this case equation (\ref{ptiineccon}) reduces to the constraint
\begin{equation}
2\Psi_{1}^{2}\Psi_{4}-\Psi_{2}^{3}=0,    \label{ptiiicon2}
\end{equation}
\noindent and since $\Psi_{4}^{(1)}=0$ we obtain the null rotation
parameter, $c=-\Psi_{2}/(2\Psi_{1})$.  Therefore,
$\Psi_{2}^{(1)}=0$ with $\Psi_{1}^{(1)}$ as the only remaining
nonzero Weyl scalar; thus ${\bf k}={\bf n}^{(1)}$ is the WAND for
a Petrov type III Weyl tensor in the kinematic frame
${\bl}^{(1)}={\bl}$.

\noindent \emph{Case 2}: $3\Psi_{2}^2-4\Psi_{1}\Psi_{3}\neq 0$.  The null rotation parameter, $c$, is the root of (\ref{ptiineccon}) which upon substitution into (\ref{psi41}) and (\ref{ptiikin}) results in
\begin{equation}
F:=27\Psi_{1}^2\Psi_{4}^2-108\Psi_{1}\Psi_{2}\Psi_{3}\Psi_{4}+64\Psi_{1}\Psi_{3}^3+54\Psi_{4}\Psi_{2}^3 -36\Psi_{3}^2\Psi_{2}^2=0.     \label{Fcon}
\end{equation}
\noindent If there exists a Petrov type II WAND ${\bf
k}={\bl}^{(3)}\neq {\bl}$, then in the kinematic frame the Weyl
scalars satisfy (\ref{Fcon}) which is equivalent to the Weyl
aligned frame condition (\ref{ptii3}).  Since there was no
remaining freedom left at equation (\ref{ptii3}) it must
invariantly define Petrov type II; a similar conclusion holds for
$F=0$.  More precisely, in the kinematic frame we have a
factorization of the invariant
\begin{equation}
I^3 - 27J^2 = -\Psi_{1}F=0,    \label{ijf}
\end{equation}
\noindent whose vanishing is the well-known result that the Weyl
tensor is Petrov type II or D (assuming $I$ and $J$ are nonzero).
By assuming $\Psi_{1}\neq 0$, it was possible to normalize
$\Psi_{1}^{(2)}$ in transformation 2).  However, it is known
\cite{kramer} that in a NP tetrad where $\Psi_{0}=\Psi_{1}=0$ the
Weyl tensor is Petrov type D if and only if
$3\Psi_{2}\Psi_{4}=2\Psi_{3}^2$ is satisfied.  Evidently, assuming
$\Psi_{1}\neq 0$ implicitly excludes Petrov type D.

In addition to a Weyl tensor of Petrov type II we also require
Riemann type II, which implies that the Ricci tensor is PP-type II
or less and the Weyl and Ricci aligned frames coincide but differ
from the kinematic frame.  An invariant characterization of a
PP-type II Ricci tensor is given by the following syzygy
\begin{equation}
r_{2}^2(4r_{1}^3-6r_{1}r_{3}+r_{2}^2)-r_{3}^2(3r_{1}^2-4r_{3})=0,    \label{rictypeii}
\end{equation}
\noindent which is expressed in terms of the Carminati-Zakhary
(CZ) Ricci invariants \cite{Carminati} (see definitions later).
Thus if Riemann type II, then in the kinematic frame (\ref{ijf})
and (\ref{rictypeii}) are satisfied and the alignment of the Weyl
and Ricci frames will impose further constraints relating the Weyl
and Ricci scalars through syzygies among the mixed invariants.  A
solution with the property that the Riemann type II aligned frame
does not coincide with the kinematic frame can be found by noting
that $\Psi_{1}\neq 0$ implies $\Phi_{01}\neq 0$; we can thus
regard (\ref{ijf}) and (\ref{rictypeii}) as polynomials in
$\Psi_{1}$, $\Phi_{01}$ and $\Phi_{10}$ with constraints given by
the vanishing of their coefficients.  We obtain
$\Psi_{3}=\Psi_{4}=0$ and $\Phi_{12}=\Phi_{22}=0$ from equations
(\ref{ijf}) and (\ref{rictypeii}), respectively. Consequently, the
Riemann tensor has vanishing negative boost weight components and
${\bf n}$ is the Riemann aligned null direction in the kinematic
frame.  If the null congruence defined by ${\bf n}$ is geodesic,
expansion-free, shear-free and twist-free, then by applying null
rotations in the following order 3) $b=i$, 2) $c=-i$ and 3) $b=i$
results in a NP tetrad ${\bl}'={\bf n}$, ${\bf n}'={\bl}$, ${\bf
m}'={\bf \overline{m}}$, ${\bf \overline{m}}'={\bf m}$. Thereby
interchanging ${\bl}$ and ${\bf n}$ the Riemann tensor has
vanishing positive  boost weight components which gives a Riemann
type II aligned null direction in the kinematic frame defined by
${\bl}'$.  To avoid this possibility, we require that at least one
of the kinematic scalars associated with the null congruence
defined by ${\bf n}$ is non-vanishing.  ${\bf n}$ is geodesic if
$\nu =0$ and affinely parameterized if $\gamma+\overline\gamma
=0$, shear-free if $\lambda=0$, expansion-free and twist-free if
the real and imaginary parts of $\mu$ vanish, respectively.
Therefore, $\nu$, $\lambda$, $\mu$ are the analogs of $\kappa$,
$\sigma$, $\rho$ for ${\bf n}$.  In the next subsection
\ref{subcases} we provide an example of this type of solution.

\subsubsection{Subcases:}\label{subcases}

There are a number of algebraically special (Riemann) subcases.
First, there is the algebraically special Riemann type II subcase,
in which there is a frame in which all positive  boost weight
components of the Riemann tensor are zero. Second, there is the
aligned subcase (of this subcase) in which the frame in which
positive boost weight terms are zero and the kinematic frame are
aligned.

This is a distinct subcase to the first subcase because there
exist Algebraically special Riemann type II Kundt spacetimes which
are not aligned (and occurs when the metric function $W_{,vv} = 0$
in the kinematic frame). For example, suppose that ${\bf k}$ is a
shear-free, geodesic Riemann type I null vector, and that ${\bf
n}$ is an aligned Riemann type II null vector. Suppose that ${\bf
k}$ and ${\bf n}$ are {\em not aligned} (i.e., subcase (1) but not
subcase (2)). Choose a null frame (${\bf k}$, ${\bf n}$, ${\bf
m_1}$,${\bf m_2}$); in this frame the Riemann tensor has no boost
weight +2  and no boost weight -1 and -2 terms. It is possible to
show there are non-trivial spacetimes in this class (satisfying
the Bianchi identities etc.), in which boost weight +1 terms are
non-zero (and the zero boost weight terms are non-zero) and ${\bf
n}$ is not geodesic, expansion-free, shear-free and twist-free.
Since the components of the Ricci tensor have positive boost
weight terms, this does not violate the Goldberg-Sachs theorem
\cite{kramer}.

\subsubsection{Example:}

We now give an example where the Riemann type II aligned frame
does not coincide with the kinematic frame, in the sense that
${\bf n}$ is the Riemann-aligned null direction having
non-vanishing expansion and shear.  We define the following NP
tetrad for the Kundt metric:
\begin{eqnarray}
{\bl}=du,\ {\bf n}=\left[H+\frac{1}{2}(W_{1}^{2}+W_{2}^2)\right]du+dv,\ {\bf m}=\frac{1}{\sqrt{2}}\left[(W_{1}-iW_{2})du-dx+idy\right]  \label{example3}
\end{eqnarray}
\noindent (where $g_{ij}=\delta_{ij}$ in this example) with metric functions
\begin{equation}
\begin{array}{ll}
\lefteqn{H=-\frac{1}{8}J^2v^4-\frac{1}{2}JF_{1}v^3-\frac{1}{2}\left[(F_{2}x+F_3)J+F_{1}^2\right]v^2
-\left[(F_{2}x+F_{3})F_{1} + F_{2} +\frac{F_{2}'}{F_{2}}\right]v} &  \\
& \mbox{} \displaystyle{-\frac{1}{2}\left[(F_2x+F_3)^2+F_{4}^{2}\right]+\frac{1}{J}\left[F_{1}'+2F_{1}F_{2}\right],} \\
W_{1}=\frac{1}{2}Jv^2 +F_1 v+F_2 x+F_3, & W_2 = F_4        \label{hwwexample3}
\end{array}
\end{equation}
\noindent where
\begin{equation}
J(u)=\exp\left[-\int\left(\frac{F_{2}'}{F_2}+3F_{2}\right)du \right]
\end{equation}
\noindent and the arbitrary functions $F_{1},\ldots,F_{4}$ only
depend on $u$. (We note that coordinate transformations can still
be used to simplify these metric functions).  ${\bl}$ has
$\kappa=\sigma=\rho=0$ and the only non-vanishing curvature
scalars are $\Psi_{1}$, $\Psi_{2}$, $\Phi_{01}$, $\Phi_{02}$,
$\Phi_{11}$ and $\Lambda$ (all negative boost weight  components
vanish) -- therefore (\ref{ijf}) and (\ref{rictypeii}) are
satisfied. Moreover, $\nu=0$. However, $\gamma +
\overline{\gamma}\neq 0$, so that ${\bf n}$ is geodesic but not
affinely parameterized. By performing a boost ${\bl}\rightarrow
A{\bl}$, ${\bf n}\rightarrow A^{-1}{\bf n}$ where
\begin{equation}
A(u)=J^{-1}\exp\left(-2\int F_{2} du\right),
\end{equation}
\noindent we can affinely parameterize ${\bf n}$ (while also
maintaining ${\bl}$ affinely parameterized).  As a result we find
that ${\bf n}$ is geodesic and twist-free; however,
$\lambda+\overline{\lambda}$ and $\mu+\overline{\mu}$ are non-zero
so that it is expanding and shearing.  We have shown that there
exists a Riemann type II aligned frame, with aligned null
direction ${\bf n}$,  that differs from the kinematic frame
determined by ${\bl}$.

Note that there are further subcases depending upon whether
$\nabla(Riem)$ is algebraically special and aligned. We shall discuss this
in detail in the next section.
Essentially, for the `zeroth order' degenerate Kundt spacetimes (with
$W_{,vv} = 0$),  in the kinematic frame   $\nabla(Riem)$ can have positive  boost weight terms.
In particular, the boost weight +1 component $R_{0101;1} \sim D\Phi_{11} \sim H_{,vvv}$
is non-zero in general. However, if
$H_{,vvv}= 0$, $\nabla(Riem)$ has no positive  boost weight terms in the kinematic frame
and the Kundt spacetime is `first order' degenerate.


\section{Degenerate Kundt spacetimes}

In the aligned-algebraically special Riemann type II subcase the
boost weight +1 terms of the Riemann tensor are also zero (in the
kinematic frame). In the analysis of \cite{inv} it was found that
a Kundt metric is $\mathcal{I}$-non-degenerate if the metric
functions $W(v,u,x)$ and $H(v,u,x)$ in the kinematic Kundt null
frame satisfy $W_{,vv} \neq 0$ and $H_{,vvv} \neq 0$. Hence the
exceptional spacetimes are the  {\it aligned algebraically special
Riemann type II Kundt spacetimes} and, from the fact that
$H_{,vvv} = 0$, in this frame $\nabla(Riem)$ (and, as we shall
show later, all covariant derivatives) do not have any positive
boost weight terms. In short (and consistent with the terminology
of the above theorem), we shall call such spacetimes {\it
degenerate} Kundt spacetimes, in which there exists a common null
frame in which the geodesic, expansion-free, shear-free and
twist-free null vector $\ell$ is also the null vector in which all
positive  boost weight terms of the Riemann tensor and its
covariant derivatives are zero (i.e., the kinematic Kundt frame
and the Riemann type II aligned null frame and the
$\nabla(Riemann)$ type II aligned null frame are all {\it
aligned}). The {\em degenerate} Kundt spacetimes are the only
spacetimes in 4D that are {\em not} $\mathcal{I}$-non-degenerate,
and their metrics are the \emph{only metrics not determined by
their curvature invariants} \cite{inv}. We note that the important
CSI and VSI spacetimes are degenerate Kundt spacetimes. We also
note that the degenerate Kundt spacetimes are the original 4D
spacetimes satisfying the energy conditions and for particular
physically motivated energy-momentum tensor (e.g.,
Einstein-Maxwell fields, with a restricted Ricci type satisfying
$W_{,vv} = 0$ and $H_{,vvv} = 0$) \cite{kramer} studied by Kundt
\cite{Kundt}

\subsection{Analysis}
Consider the 4D Kundt metric (\ref{Kundt}) and a kinematically
aligned NP tetrad, so that ${\bl}=du$ has
$\kappa=\sigma=\rho=\epsilon=0$.  In order to delineate among the
Kundt metrics the ones that are $\mathcal{I}$-non-degenerate we
make the following distinction:

\begin{defn}
Let $K_{n}$ denote the subclass of Kundt metrics such that there
exists a kinematic frame in which Riemann up to and including its
$n^{th}$ covariant derivative have vanishing positive  boost
weight  components.  We call $K_{n}$ the $n^{th}$ order degenerate
Kundt class.
\end{defn}

\noindent Therefore, for every metric in $K_{n}$ there exists an
NP tetrad in which $\kappa=\sigma=\rho=\epsilon=0$ and
$R_{abcd},\ldots,\nabla^{(n)}R_{abcd}$ are type II or less.  In
general the Kundt metric (\ref{Kundt}) is Riemann type I; imposing
$K_{0}$ is equivalent to $\Psi_{1}=\Phi_{01}=0$, which in terms of
metric functions gives $W_{i,vv}=0$.  It now follows that $K_{0}$
does not imply $K_{1}$ since $\nabla R_{abcd}$ has a boost weight
+1 component linear in the scalars $D\Psi_{2}$, $D\Phi_{11}$ and
$D\Lambda$.  Requiring $K_{1}$ is equivalent to $H_{,vvv}=0$ (in
addition to $W_{i,vv}=0$) and consequently
$D\Psi_{2}=D\Phi_{11}=D\Lambda=0$.  A direct calculation now shows
that $K_{1}$ implies $K_{2}$, leading us to the following result
(the converse is trivial by definition of $K_n$). {\footnote{ A
similar theorem is valid in arbitrary dimensions -- see later. }}

\begin{thm}\label{thmK2}
In the Kundt class, $K_{1}$ implies $K_{n}$ for all $n\geq 2$.
\end{thm}
\begin{proof}
In 4 dimensions the transverse metric of (\ref{Kundt}) is
conformally flat so that $g_{ij}dx^i dx^j = -P^{-2}(dx^2+dy^2)$,
where $P=P(u,x,y)$ \cite{kramer}.  We choose an NP tetrad of the
form
\begin{eqnarray}
{\bl}=du,\ {\bf n}=\left[H+\frac{P^2}{2}(W_{1}^{2}+W_{2}^2)\right]du+dv,\ {\bf m}=\frac{P}{\sqrt{2}}\left[(W_{1}-iW_{2})du-\frac{(dx-i dy)}{P^2}\right],   \label{npframe2}
\end{eqnarray}
\noindent and impose $K_1$ to find that the boost weight 0
components of $R_{abcd;e}$ all contain $W_{i}$ with at least one
derivative with respect to $v$, and $H$ with at least two
derivatives with respect to $v$.  In addition, the boost weight -1
components have $H$ appearing with at least one derivative with
respect to $v$.  It now follows by calculation that $R_{abcd;ef}$
satisfies $K_2$ (positive  boost weight components vanish) and its
boost weight 0 and -1 components have the same dependence as
$R_{abcd;e}$ with respect to the derivatives of $W$ and $H$ with
respect to $v$.  Proceeding by induction, we consider the $n^{th}$
covariant derivative of Riemann and denote its components of boost
weight $b$ by $(\nabla^{n}R)_{b}$.  Assume that $K_n$ holds with
the property that the occurrence of $W_{i}$ and $H$ in the boost
weight 0 components, $(\nabla^{n}R)_{0}$, only depends
\footnote{Although the components also depend on $P(u,x,y)$, this
is of no consequence since \mbox{$DP=0$} and hence $P$ will not
appear in a boost weight +1 component.} on $W_{i,v\cdots}$ and
$H_{,vv\cdots}$ (where $\cdots$ refers to derivatives with respect
$u$, $x$ or $y$) and $(\nabla^{n}R)_{-1}$ components have a
dependence on $H_{,v\cdots}$ but not $H$.

Before considering the $n+1$ derivative of Riemann, we note that
since $\kappa=\sigma=\rho=\epsilon=0$ for (\ref{npframe2}) then
the covariant derivatives of the NP tetrad, sorted according to
decreasing boost weights, are
\begin{eqnarray}
\ell_{a;e} & = &  -(\bar{\tau}m_{a}+\tau\bar{m}_{a})\ell_{e}-(\bar{\alpha}+\beta)\ell_{a}\bar{m}_{e}- (\alpha+\bar{\beta})\ell_{a}m_{e}+ (\gamma+\bar{\gamma})\ell_{a}\ell_{e}\, , \hspace{1cm}  \label{npcdl}\\
n_{a;e} &=& (\pi m_{a}+\bar{\pi}\bar{m}_{a})n_{e} + (\bar{\alpha}+\beta)n_{a}\bar{m}_{e}+ (\alpha+\bar{\beta})n_{a}m_{e}-(\gamma+\bar{\gamma})n_{a}\ell_{e}   \nonumber \\
& & \mbox{} - (\mu m_{a}+\bar{\lambda}\bar{m}_{a})\bar{m}_{e} - (\bar{\mu}\bar{m}_{a}+\lambda m_{a})m_{e} + (\nu m_{a}+ \bar{\nu}\bar{m}_{a})\ell_{e}\, ,  \label{npcdn}\\
m_{a;e} &=& \bar{\pi}\ell_{a}n_{e}-\tau n_{a}\ell_{e}+ (\bar{\alpha}-\beta)m_{a}\bar{m}_{e}- (\alpha-\bar{\beta})m_{a}m_{e}   \nonumber \\
& & \mbox{}+(\gamma-\bar{\gamma})m_{a}\ell_{e}-\bar{\lambda}\ell_{a}\bar{m}_{e}-\bar{\mu}\ell_{a}m_{e}+ \bar{\nu}\ell_{a}\ell_{e}\, .   \label{npcdm}
\end{eqnarray}
\noindent Therefore, the covariant derivative of \emph{any} outer
product of NP tetrad vectors gives components of equal or lesser
boost weight.  Suppose an outer product has boost weight $b$; then
applying a covariant derivative gives components of boost weight
$b$ depending on $\{\tau,\pi,\alpha,\beta\}$, $b-1$ depending on
$\{\gamma,\mu,\lambda\}$ and $b-2$ depending on $\{\nu\}$, along
with their complex conjugates.

Consider the $n+1$ covariant derivative of Riemann. Then the boost
weight +1 components, $(\nabla^{n+1}R)_{+1}$, can only arise by
applying $D=\partial_{v}$ to the components of
$(\nabla^{n}R)_{0}$.  By hypothesis $(\nabla^{n}R)_{0}$ only
depends on $W_{i,v\cdots}$ and $H_{,vv\cdots}$; therefore, from
the $K_1$ conditions, $D^{2}W_{i}=D^{3}H=0$,
$(\nabla^{n+1}R)_{+1}$ vanishes whereby we obtain that the $n+1$
covariant derivative of Riemann is type II.

Furthermore, the boost weight 0 components of
$(\nabla^{n+1}R)_{0}$ arise from applying $\delta$ or
$\overline{\delta}$ to the components of $(\nabla^{n}R)_{0}$.  By
assumption $(\nabla^{n}R)_{0}$ has the property that they only
depend on $W_{i,v\cdots}$ and $H_{,vv\cdots}$. Since
$\delta=P(\partial_{x}-i\partial_{y})/\sqrt{2}$ then
$(\nabla^{n+1}R)_{0}$ will also depend on $W_{i,v\cdots}$ and
$H_{,vv\cdots}$\, .  In addition, $(\nabla^{n+1}R)_{0}$ also arise
from $(\nabla^{n}R)_{0}$ by taking the covariant derivative of the
outer product of NP tetrad vectors whose total boost weight is
zero.  From (\ref{npcdl})--(\ref{npcdm}) we showed that these
contributions will have the form of $\tau$, $\pi$, $\alpha$ or
$\beta$ multiplied by the components of $(\nabla^{n}R)_{0}$. Since
these spin coefficients only depend on $W_{i,v}$ (and $P$) it
follows that the components of $(\nabla^{n+1}R)_{0}$ depend on
$W_{i,v\cdots}$ and $H_{,vv\cdots}$\, .  Lastly,
$(\nabla^{n+1}R)_{0}$ also have contributions from applying $D=
\partial_{v}$ to $(\nabla^{n}R)_{-1}$.  Since it is assumed that $(\nabla^{n}R)_{-1}$ does not
depend on $H$ (i.e.,  H appears with at least one derivative with respect to $v$), then we
conclude that $(\nabla^{n+1}R)_{0}$ will only depend on $W_{i,v\cdots}$ and $H_{,vv\cdots}$\, .

Next consider the boost weight -1 components
$(\nabla^{n+1}R)_{-1}$ which arise from either applying $\delta$
or $\overline{\delta}$ to $(\nabla^{n}R)_{-1}$ or as products of
$\tau$, $\pi$, $\alpha$ or $\beta$ with $(\nabla^{n}R)_{-1}$.
Since $(\nabla^{n}R)_{-1}$ has $H$ appearing with at least one
derivative with respect to $v$, then from the same argument used
above in the case of $(\nabla^{n+1}R)_{0}$ we have that
$(\nabla^{n+1}R)_{-1}$ will also have the same $H_{,v\cdots}$
dependence (no dependence on $H$).  Moreover,
$(\nabla^{n+1}R)_{-1}$ also arises from applying $D$ to
$(\nabla^{n}R)_{-2}$; however, no matter what functional
dependence we have in $(\nabla^{n}R)_{-2}$ we again obtain that
$H$ appears with at least one derivative with respect to $v$ in
$(\nabla^{n+1}R)_{-1}$.  Finally, we have contributions to
$(\nabla^{n+1}R)_{-1}$ occurring through the application of
$\Delta$ to $(\nabla^{n}R)_{0}$ or as products of $\gamma$, $\mu$
or $\lambda$ with $(\nabla^{n}R)_{0}$.  In the first case, since
\begin{equation}
\Delta=\partial_{u}-\left[H+\frac{P^2}{2}(W_{1}^{2}+W_{2}^{2}) \right]\partial_{v}+P^2W_{1}\partial_{x} +P^2W_{2}\partial_{y}
\end{equation}
\noindent then when applied to $(\nabla^{n}R)_{0}$ we have, by
virtue of the $K_{1}$ conditions, that $\partial_{v}$ on the
components of $(\nabla^{n}R)_{0}$ vanish.  Consequently, by
hypothesis on $(\nabla^{n}R)_{0}$ then $(\nabla^{n+1}R)_{-1}$ has
$H$ appearing with at least one derivative of $v$.  In the second
case we note that $\{\gamma,\mu,\lambda\}$ only has $H$ appearing
with one derivative of $v$ in $\gamma$, and it follows that the
product of these spin coefficients with $(\nabla^{n}R)_{0}$ will
give rise to $H$ appearing with at least one derivative of $v$ in
$(\nabla^{n+1}R)_{-1}$.

We can now conclude that the $n+1$ covariant derivative of Riemann
is type II, the boost weight  0 components, $(\nabla^{n+1}R)_{0}$,
only depend on $W_{i,v\cdots}$ and $H_{,vv\cdots}$, and the boost
weight -1 components, $(\nabla^{n+1}R)_{-1}$, contain $H$ with at
least one derivative with respect to $v$; therefore, $K_{n+1}$
holds.
\end{proof}

We note that the VSI (or proper CSI) spacetimes satisfy $K_{n}$
and have vanishing (or constant) boost weight 0 components for all
orders $n$.  This result partially characterizes the metrics not
determined by their scalar polynomial curvature invariants, namely
all $\mathcal{I}$-degenerate metrics must satisfy $K_{1}$.  We
note that there may exist cases where Riemann and all of its
covariant derivatives are type D (and hence $K_n$ for all $n$) but
a sufficient number of independent curvature invariants can be
constructed such that Riemann and its derivatives can be
determined in some sense. We return to this in the next
subsection.

We also note that the Goldberg-Sachs theorem (GS: theorem 7.1 in
\cite{kramer}) states that a spacetime with a shear-free, geodesic
null congruence ${\bf k}$ ($\kappa=\sigma=0$) satisfying
$R_{ab}k^ak^b=0$, $R_{ab}k^am^b=0$, $R_{ab}m^am^b=0$ (what we
might call a spacetime admitting an aligned algebraically special
Ricci tensor), necessarily has $\Psi_0=\Psi_1=0$ (aligned
algebraically special Weyl tensor). When applied to the Kundt
class, GS implies $K_{0}$. However, the conditions of GS on the
Ricci tensor are slightly stronger than what are required for
$K_{0}$ to be satisfied; GS also imposes
$R_{ab}m^{a}m^{b}=\Phi_{02}=0$.  Based on the known solutions
\cite{kramer} of the Kundt class satisfying GS, these conditions
and $\rho=0$ imply $K_{0}$ which in turn implies $K_{n}$ for all
$n\geq 1$.


\section{Classification of degenerate Kundt spacetimes}

The degenerate Kundt spacetimes are classified algebraically by their
Riemann type (II,III,N,D or O) in the aligned kinematic frame, or more finely by
their Ricci and Weyl types separately. We are only interested here in types II and D, since otherwise
the degenerate Kundt spacetime is VSI \cite{pravda}. Within each algebraic type, it may also be useful to
classify the covariant derivative(s) of Riemann tensor (and particularly
$\nabla(Riem)$ and $\nabla^2(Riem)$) in terms of their algebraic types.

In the analysis below it is important to fix the frame in each
algebraic class completely. We shall begin with the classification
of $\nabla(Riem)$.

\subsubsection{Classification of $\nabla(Riem)$ for Riemann type II}

The degenerate Kundt spacetimes of proper Riemann type II (i.e.,
with some non-zero boost weight zero terms but not type D) are
further classified in the aligned kinematic frame by the algebraic
type of $\nabla(Riem)$ (and, for example, $\nabla^2(Riem)$). In
general, $\nabla(Riem)$ in a degenerate Kundt spacetime is of type
(II,G) {\footnote{The notation is consistent with that of
\cite{pravda} and \cite{class}. Since, by definition,
$\nabla(Riem)$ has no positive boost weight terms, its principal
type is II or more special. In general its secondary type is G,
but if there are no boost weight -3 (-2) terms it is of type H
(I), etc. If $\nabla(Riem)$ is of type (II,II), and consequently
has only boost weight zero terms, it is said to be of type D.}}.
In Table \ref{algtypestable} we list the possible algebraic types
of $\nabla(Riem)$ corresponding to the vanishing of the
appropriate boost weight  components in the degenerate Kundt
class. Further subcases consist of types (II,H), (II,I), (II,II),
and (III/N,H), (III/N,I) and (III/N,II), etc. Further subclasses
exist, including those where certain contractions of
$\nabla^{n}(Riem)$ are separately classified algebraically. The
conditions for each of these subclasses can be presented in a
similar manner to those of the next section.

\begin{table}[h]
\begin{tabular}{c|ccccccccccc}
bw  & (II,G) & (II,H) & (II,I) & D      & (III,G) & (III,H) & (III,I) &
(N,G)  & (N,H)  & (O,G)  & O  \\ \hline
0   & $\ast$ & $\ast$ & $\ast$ & $\ast$ & 0       & 0       &   0     &
0    &   0    &   0    & 0  \\
-1  & $\ast$ & $\ast$ & $\ast$ & 0      & $\ast$  & $\ast$  & $\ast$  &
0    &   0    &   0    & 0  \\
-2  & $\ast$ & $\ast$ &   0    & 0      & $\ast$  & $\ast$  &   0     &
$\ast$ & $\ast$ &   0    & 0  \\
-3  & $\ast$ &   0    &   0    & 0      & $\ast$  &   0     &   0     &
$\ast$ &   0    & $\ast$ & 0  \\
\end{tabular}
\caption{Within the the degenerate Kundt class we list the
possible algebraic types of $\nabla(Riem)$ corresponding to the
vanishing of the appropriate boost weight (bw) components.  In
each of these types it is assumed that either there is no
remaining frame freedom or $Riem$ and $\nabla(Riem)$ have an
isotropy group consisting of 2-dimensional null rotations about
$\ell$ (or subgroup thereof). A simplified notation, as defined in
the text, is often used; for example,  D:=(II,II), III:=(III,G) or
(III,H) and N:=(N,G) or(N,H).} \label{algtypestable}
\end{table}

\subsubsection{Classification of $\nabla(Riem)$ type III, N or O}

Of particular interest, perhaps, is the case where $\nabla(Riem)$ (and $\nabla^{n}(Riem)$)
are of type III, N or O, and hence all terms in $\nabla(Riem)$ ($\nabla^{n}(Riem)$) are of
negative boost weight and hence do not contribute to any scalar invariants containing
covariant derivative(s) of the Riemann tensor (i.e., the only non-vanishing polynomial
scalar invariants are the zeroth order ones constructed from the Riemann tensor alone).

Let us present the calculation of Riemann type II and
$\nabla(Riem)$ of type (III,G) or more specialized. We begin by
assuming there exists an NP tetrad in which Riemann is type II
(therefore $\Psi_{0}=\Psi_{1}=\Phi_{00}=\Phi_{01}=0$) and
simultaneously in this frame $\kappa=\sigma=\rho=\epsilon=0$.
Recall that this is the definition of a zeroth order degenerate
Kundt spacetime which we denoted as $K_0$.  The remaining frame
freedom consists of a 2-dimensional group of null rotations about
$\ell$, to be used later to simplify the Weyl tensor and thus
completely fix the frame.  Calculating $\nabla(Riem)$ and
employing the Bianchi identities results in components of boost
weight  $\leq +1$. Since this case considers $\nabla(Riem)$ of
type III or less we require the vanishing of boost weight  +1 and
0 components in this frame.  All components of $\nabla(Riem)$ of
boost weight  +1 reduce to a constant multiple of $DR$; hence a
necessary and sufficient condition for $\nabla(Riem)$ to be type
II is $DR=0$.  The Bianchi equations provide further conditions
showing that $D$ of all boost weight  0 components of $Riem$
vanish:

\noindent $(\nabla R)_{+1}=0$:
\begin{equation}
D\Psi_{2}=D\Phi_{11}=D\Phi_{02}=DR=0 \, .   \label{Dbw01ii}
\end{equation}
\noindent Note that the boost weight  +1 components of $\nabla(Riem)$ are invariant with
respect to null rotations about $\ell$; therefore, the conditions (\ref{Dbw01ii}) are
invariant conditions with respect to this remaining frame freedom.
Assuming boost weight  +1 components of $\nabla(Riem)$ vanish, the necessary and
sufficient conditions for the boost weight  0 components to vanish are then:

\noindent $(\nabla R)_{0}=0$:
\begin{eqnarray}
& \tau\Psi_{2}=0 &  \label{bw00iii}  \\
& 2\tau\Phi_{11}+\bar{\tau}\Phi_{02}=0 &    \label{bw01iii}  \\
& D\Psi_{3}=3\pi\Psi_{2},\quad D\Phi_{12}=2\bar{\pi}\Phi_{11}+\pi\Phi_{02}
&  \label{bw02iii}  \\
& \delta R=0 &  \label{bw03iii}  \\
& \bar{\delta}\Phi_{02}=2(\alpha-\bar{\beta})\Phi_{02},\quad
\delta\Phi_{02}=-2(\bar{\alpha}-\beta)\Phi_{02} \, .& \label{bw04iii}
\end{eqnarray}
\noindent Again, (\ref{bw00iii})--(\ref{bw04iii}) follow after
(\ref{Dbw01ii}) has been imposed.  As above, if the boost weight  +1 components vanish,
then the boost weight  0 components are invariant with respect to null rotations
about $\ell$; hence (\ref{bw00iii})--(\ref{bw04iii}) provide invariant
conditions for the vanishing of the boost weight  0 components.  The remaining
non-vanishing components of $\nabla(Riem)$ have boost weights  -1, -2 and -3.
Therefore, in a $K_0$ spacetime, $\nabla(Riem)$ is of type (III,G) if and only
if equations (\ref{Dbw01ii}) and (\ref{bw00iii})--(\ref{bw04iii}) hold.
Consequently, the Bianchi equations give some additional useful relations:
$\delta\Psi_{2}=\bar{\delta}\Psi_{2}=\delta\Phi_{11}=0$.

We now specialize by assuming that Weyl is proper type II; i.e.,  $\Psi_{2} \neq 0$  and
the type does not
reduce to another more algebraically special type.  Using a
null rotation about $\ell$ to set $\Psi_{3}=0$, the frame is fixed with
$\Psi_{2}$ and $\Psi_{4}$ nonzero.  If $\Psi_{3}$ was zero initially, then
the frame is already fixed since any non-trivial null rotation about
$\ell$ results in a nonzero $\Psi_{3}$.  If under the null rotation about
$\ell$ both $\Psi_{3}$ and $\Psi_{4}$ become zero then we obtain Weyl type
D, a case which is excluded since we have assumed proper type II.  In this
fixed NP tetrad, equations (\ref{Dbw01ii}) and
(\ref{bw00iii})--(\ref{bw04iii}) simplify, giving a number of cases.
Here, we consider one of these:

\noindent \emph{Case 1}: $\tau=0$, $\Psi_{2}\neq 0$, $\Psi_{3}=0$,
$\Psi_{4}\neq 0$.  $\nabla(Riem)$ is type (III,G) if and only if
\begin{eqnarray}
& \tau=\pi=0 &  \label{t001}  \\
& \Psi_{2}-\displaystyle{\frac{R}{12}}=\Phi_{02}=0 &   \label{t002}  \\
& D\Psi_{2}=D\Phi_{11}=D\Phi_{12}=DR=0 &    \label{t003}  \\
& \delta\Psi_{2}=\delta\Phi_{11}=\delta R =0\, . &  \label{t004}
\end{eqnarray}
\noindent Note that the first equation of (\ref{t002}) implies
that $\Psi_{2}$ is real and, in addition,  since $\Psi_{2}\neq 0$
in this case, $R \neq 0$.  Further algebraically special types of
$\nabla(Riem)$ assume (\ref{t001})--(\ref{t004}) are satisfied.
Implementing the NP and Bianchi equations, we consider the
vanishing of $\nabla(Riem)$ components at each boost order
separately.

\noindent $(\nabla R)_{-1}=0$:
\begin{eqnarray}
& \mu=\lambda=0 &  \label{bwm10iii}  \\
& D\Psi_{4}=D\Phi_{22}=0 &    \label{bwm11iii}  \\
& \Delta R=0 &  \label{bwm12iii}  \\
& D\nu=\Phi_{12}=0 &  \label{bwm13iii}  \\
& R,\, \Psi_{2},\, \Phi_{11} \text{ are constants.} & \label{bwm14iii}
\end{eqnarray}
\noindent In (\ref{bwm13iii}) and (\ref{bwm14iii}) we list some of
the consequences of the NP and Bianchi equations. In particular,
(\ref{bwm14iii}) implies that all zeroth order scalar polynomial
invariants of the Riemann tensor are constant.

\noindent $(\nabla R)_{-2}=0$:
\begin{eqnarray}
& \nu=0 &  \label{bwm20iii}  \\
& \delta\Phi_{22}=-2(\bar{\alpha}+\beta)\Phi_{22} &    \label{bwm21iii}  \\
& \bar{\delta}\Psi_{4}=-4\alpha\Psi_{4},\quad
\delta\Psi_{4}=-4\beta\Psi_{4} &  \label{bwm22iii}  \\
& \Phi_{12}=0 &  \label{bwm23iii}
\end{eqnarray}

\noindent $(\nabla R)_{-3}=0$:
\begin{eqnarray}
&
\Delta\Phi_{22}=-2(\gamma+\bar{\gamma})\Phi_{22}+2\nu\Phi_{12}+2\bar{\nu}\Phi_{21}
& \label{bwm31iii} \\
& \Delta\Psi_{4}=-4\gamma\Psi_{4} & \label{bwm32iii}
\end{eqnarray}

\noindent Therefore, in this Case 1, $\nabla(Riem)$ is at most
type (III,G) and is of more special type if the following
conditions are satisfied:
\begin{itemize}
\item[(III,H) :] (\ref{bwm31iii}), (\ref{bwm32iii})
\item[(III,I) :] (\ref{bwm31iii}), (\ref{bwm32iii}) and
(\ref{bwm20iii})--(\ref{bwm23iii}) ; here (\ref{bwm31iii}) simplifies to
$\Delta\Phi_{22}=-2(\gamma+\bar{\gamma})\Phi_{22}$.
\item[(N,G) :] (\ref{bwm10iii})--(\ref{bwm14iii})
\item[(N,H) :] (\ref{bwm10iii})--(\ref{bwm14iii}) and (\ref{bwm31iii}),
(\ref{bwm32iii}) ; here (\ref{bwm31iii}) simplifies to
$\Delta\Phi_{22}=-2(\gamma+\bar{\gamma})\Phi_{22}$.
\end{itemize}
\noindent  For algebraic type (O,G), where
(\ref{bwm10iii})--(\ref{bwm23iii}) hold, it follows that $\nu=\mu=\lambda=0$.
Therefore, in addition to $\Phi_{12}=0$, the NP equations give
$\Psi_{4}=\Phi_{22}=0$; however, in Case 1 $\Psi_{4}\neq 0$ and thus type (O,G)
is excluded.  By a similar argument type O is also excluded; i.e., a
symmetric space cannot exist in Case 1.

Since $\nabla(Riem)$ is at most type (III,G) it has only negative boost
weight components; therefore, all first order scalar polynomial invariants
vanish.  Taking another covariant derivative we observe that the boost weight  0
components of $\nabla^{(2)}(Riem)$ can only arise from taking $D$ of the
boost weight  -1 components of $\nabla(Riem)$.  After simplifying we notice that
every boost weight  -1 component of $\nabla(Riem)$ is a sum of the terms $D\Psi_{4}$,
$D\Phi_{22}$, $\Delta R$, $\mu R$ or $\lambda R$.  Using the NP equations
\begin{equation}
D\lambda=D\mu=D\alpha=D\beta =0,  \label{NPeqnszero}
\end{equation}
and (\ref{t003}), we have that $D(\mu R)=D(\lambda R)=0$.  In addition, consider
the following Bianchi equations and commutation relations:
\begin{eqnarray}
& D\Psi_{4}=\bar{\delta}\Phi_{21}-\displaystyle{\frac{\lambda
R}{4}}-2\lambda\Phi_{11}+2\alpha\Phi_{21} & \label{bianDps4} \\
&
D\Phi_{22}=\bar{\delta}\Phi_{12}-\displaystyle{\frac{\bar{\mu}R}{4}}-2\bar{\mu}\Phi_{11}+2\bar{\beta}\Phi_{12}
& \label{bianDph22} \\
& [\Delta,D]=(\gamma+\bar{\gamma})D,\quad
[\bar{\delta},D]=(\alpha+\bar{\beta})D \, . &  \label{commDeldelD}
\end{eqnarray}
\noindent The first equation of (\ref{commDeldelD}) and equation (\ref{t003}) give
$D\Delta R =0$.  Next, we take $D$ of (\ref{bianDps4}) and (\ref{bianDph22})
and apply the second equation of (\ref{commDeldelD}), (\ref{NPeqnszero})
and (\ref{t003}) to obtain $D^{2}\Psi_{4}=D^{2}\Phi_{22}=0$.  Therefore all
boost weight  0 components of $\nabla^{(2)}(Riem)$ vanish, and the algebraic type
is at most (III,G) and all second order scalar polynomial invariants
vanish.  We expect that higher order covariant derivatives of the Riemann
tensor will also have vanishing boost weight  0 components so that
$\nabla^{(n)}(Riem)$ is type (III,G) for all $n\geq 1$, thus giving an
$\mathcal{I}$-symmetric space \cite{inv}.

\subsection{Classification of $\nabla(Riem)$: type D}

In order to perform an algebraic classification of $\nabla(Riem)$
within the degenerate Kundt class we shall begin by choosing to
use the NP tetrad where the kinematic frame and the Riemann
aligned frame coincide.  In particular,
$\Psi_{0}=\Psi_{1}=\Phi_{00}=\Phi_{01}=0$ so that Riemann is type
II or less and in this frame $\kappa=\sigma=\rho=\epsilon=0$.
Moreover, since degenerate Kundt has $\nabla^{(n)}(Riem)$ type
(II,G) or less for all $n\geq 0$, one consequence was found to be
$D\Psi_{2}=D\Phi_{11}=DR=0$.  By fixing $\epsilon=0$ using a
spin-boost then the only remaining freedom preserving the
kinematic and Riemann aligned frames are null rotations about
$\ell$.  The classification of $\nabla(Riem)$ will, in general, be
subject only to this 2-dimensional group of null rotations.

\subsubsection{Riemann type D}

The only non-vanishing curvature scalars are $\Psi_{2}$,
$\Phi_{11}$, $\Phi_{02}$ and $R$, and hence Riemann has only boost
weight 0 components.  Since any further null rotations about
$\ell$ do not preserve the Riemann type D aligned frame, the frame
is thus completely fixed and the classification of $\nabla(Riem)$
will depend on the vanishing of its boost weight components at
various orders. For Riemann type D we find that, in general,
$\nabla(Riem)$ is type (II,H) (boost weight -3 vanish) and must
satisfy the Bianchi equations
\begin{eqnarray}
0 & = & 3\nu\Psi_{2}-2\nu\Phi_{11}-\bar{\nu}\Phi_{20}   \label{bianbm21} \\
\Delta\Psi_{2}-\frac{1}{12}\Delta R & = & -3\mu\Psi_{2}-2\mu\Phi_{11}-\bar{\lambda}\Phi_{20}    \label{bianbm11} \\
\Delta\Phi_{20} & = & -3\lambda\Psi_{2}-2\lambda\Phi_{11}-\bar{\mu}\Phi_{20}-2(\gamma-\bar{\gamma})\Phi_{20} \label{bianbm12} \\
\Delta\Phi_{11}-\frac{1}{8}\Delta R & = & -2(\mu+\bar{\mu})\Phi_{11}-\lambda\Phi_{02}-\bar{\lambda}\Phi_{20} \label{bianbm13} \\
\delta\Psi_{2}-\bar{\delta}\Phi_{02}-\frac{1}{12}\delta R & = & 3\tau\Psi_{2}-2\tau\Phi_{11}-(2\alpha+\bar{\tau}- 2\bar{\beta})\Phi_{02}          \label{bianb01} \\
\bar{\delta}\Psi_{2}-\delta\Phi_{20}-\frac{1}{12}\bar{\delta}R & = & -3\pi\Psi_{2}+2\pi\Phi_{11}-(2\bar{\alpha} -\bar{\pi}-2\beta)\Phi_{20}     \label{bianb02} \\
\delta\Phi_{11}+\bar{\delta}\Phi_{02}+\frac{1}{8}\delta R & = & 2(\tau-\bar{\pi})\Phi_{11}+(\bar{\tau}-\pi + 2\alpha -2\bar{\beta})\Phi_{02} \, .    \label{bianb03}
\end{eqnarray}
\noindent Calculating the components of $\nabla(Riem)$ shows that
the Bianchi equations (\ref{bianbm21}),
(\ref{bianbm11})--(\ref{bianbm13}) and
(\ref{bianb01})--(\ref{bianb03}) provide constraints for the boost
weight -2, -1 and 0 components, respectively.  Our analysis takes
into account the Bianchi equations; however, we shall not consider
the NP equations nor the commutation relations.  In principle,
these would have to be solved in order to obtain the class of
metrics corresponding to Riemann type D with a given algebraic
type of $\nabla(Riem)$.  As an example, using Bianchi equation
(\ref{bianbm21}) all boost weight -2 components of $\nabla(Riem)$
are proportional to $\nu\Psi_{2}$.  Solving for the vanishing of
all components of $\nabla(Riem)$ at each boost order separately
gives:

\noindent $(\nabla R)_{-2}=0$:
\begin{eqnarray}
& \nu\Psi_{2}=0 &  \label{nabrm21} \\
& 2\nu\Phi_{11}+\bar{\nu}\Phi_{20}=0 \, .&    \label{nabrm22}
\end{eqnarray}
\noindent $(\nabla R)_{-1}=0$:
\begin{eqnarray}
& \Delta\Psi_{2}=\Delta\Phi_{11}=\Delta R = 0  &  \label{nabrm11} \\
& \lambda\Psi_{2}=0,\quad \mu\Psi_{2}=0    &  \label{nabrm12}  \\
& 2\bar{\mu}\Phi_{11}+\lambda\Phi_{02}=0,\quad 2\bar{\lambda}\Phi_{11}+\mu\Phi_{02}=0 & \label{nabrm13} \\
& \Delta\Phi_{02}=2(\gamma-\bar{\gamma})\Phi_{02} \, .&  \label{nabrm14}
\end{eqnarray}
\noindent $(\nabla R)_{0}=0$:
\begin{eqnarray}
& \delta\Psi_{2}=\bar{\delta}\Psi_{2}=\delta\Phi_{11}=\delta R=0 & \label{nabr01} \\
& \pi\Psi_{2}=0,\quad \tau\Psi_{2}=0 & \label{nabr02} \\
& 2\bar{\pi}\Phi_{11}+\pi\Phi_{02}=0,\quad 2\tau\Phi_{11}+\bar{\tau}\Phi_{02} =0 & \label{nabr03} \\
& \bar{\delta}\Phi_{02}=2(\alpha-\bar{\beta})\Phi_{02},\quad \delta\Phi_{02}=-2(\bar{\alpha}-\beta)\Phi_{02} \, .& \label{nabr04}
\end{eqnarray}
\noindent  Here we have also included the reduced Bianchi
equations relevant to each boost weight.  Evidently, from
equations (\ref{nabrm21}), (\ref{nabrm12}) and (\ref{nabr02}) the
algebraic classification of $\nabla(Riem)$ depends on whether
$\Psi_{2}$ vanishes or not. If $\Psi_{2}=0$, then equations
(\ref{nabrm22}), (\ref{nabrm13}) and (\ref{nabr03}) give rise to
additional subcases dependent on whether $4\Phi_{11}^{\ \  2} -
\Phi_{02}\Phi_{20}$ vanishes or not, thus determining if there
exists non-trivial algebraic solutions of the spin coefficients.
Since the boost weight -3 components of $\nabla(Riem)$ vanish, we
use the simplified notation for the algebraic types given below
by: III:=(III,H) and N:=(N,H).

\noindent \emph{Case 1}: $\Psi_{2}\neq 0$.  $\nabla(Riem)$ is type
\begin{itemize}
\item[(II,I) :] $\nu=0$
\item[D :] $\nu=\lambda=\mu=0$, and (\ref{nabrm11}) and (\ref{nabrm14})
\item[III :] $\tau=\pi=0$, and (\ref{nabr01}) and (\ref{nabr04})
\item[ N :] $\lambda=\mu=\tau=\pi=0$, and (\ref{nabrm11}), (\ref{nabrm14}), (\ref{nabr01}) and (\ref{nabr04})
\item[ O :] same as type N and also $\nu=0$
\end{itemize}

\noindent \emph{Case 2}: $\Psi_{2}=0$ (conformally flat).  In this
case boost weight -2 components of $\nabla(Riem)$ always vanish;
therefore, $\nabla(Riem)$ is, in general, type (II,I).

\emph{Case 2.1}: $4\Phi_{11}^{\ \ 2}-\Phi_{02}\Phi_{20}\neq 0$.  For this subcase $\nu=0$, and the rest of the conditions for $\nabla(Riem)$ types are the same as in Case 1 except type N is excluded.

\emph{Case 2.2}: $4\Phi_{11}^{\ \ 2}-\Phi_{02}\Phi_{20}=0$.  Admits non-trivial solutions for the spin coefficients and $\nu$ always satisfies (\ref{nabrm22}).  Then $\nabla(Riem)$ is type
\begin{itemize}
\item[(II,I) :] no further conditions
\item[D :] (\ref{nabrm11}), (\ref{nabrm14}) and $\lambda$, $\mu$ satisfy (\ref{nabrm13})
\item[III :] (\ref{nabr01}), (\ref{nabr04}) and $\tau$, $\pi$ satisfy (\ref{nabr03})
\item[O :] union of type D and III conditions.
\end{itemize}

\noindent Vacuum can only occur in Case 1 whereas in 2.2 it
reduces to flat space.  All vacuum Riemann type D (hence Petrov
type D) solutions are known (see (31.41) of \cite{kramer}).  In
both Cases 1 and 2, if $\nabla(Riem)$ is type O then we obtain a
symmetric space.  It is important to note that the non-vanishing
boost weight components of a given $\nabla(Riem)$ type possess
additional constraints through Bianchi which are not included
above.  For example, in $\nabla(Riem)$ type D all boost weight 0
components are also subject to (\ref{bianb01})--(\ref{bianb03}).

\subsection{Classification of $\nabla^{2}(Riem)$}

We can also algebraically classify $\nabla^{2}(Riem)$. We shall
just present type D here.

\subsubsection{Classification of $\nabla^{(2)}(Riem)$: type D}

In principle, an algebraic classification of $\nabla^{(2)}(Riem)$
can be performed for any fixed algebraic types of $Riem$ and
$\nabla(Riem)$.  We shall present a partial classification of
$\nabla^{(2)}(Riem)$ when $Riem$ and $\nabla(Riem)$ are both of
type D.

For Case 1, that is assuming $\Psi_{2}\neq 0$,  we obtain the
following boost weight +1 and -1 components, up to a constant
factor, after using the Bianchi equations and commutator
relations:

\noindent $(\nabla^{2} R)_{+1}$:
\begin{eqnarray}
& D\pi(3\Psi_2 -2\Phi_{11})-D\bar{\pi}\Phi_{20} &  \label{nabsqr11} \\
& D\pi (3\Psi_{2}+2\Phi_{11})+D\bar{\pi}\Phi_{20} &    \label{nabsqr12}
\end{eqnarray}
\noindent $(\nabla^{2} R)_{-1}$:
\begin{eqnarray}
&
[\Delta\tau-(\gamma-\bar{\gamma})\tau](3\Psi_{2}+2\Phi_{11})+[\Delta\bar{\tau}-
(\bar{\gamma}-\gamma)\bar{\tau}]\Phi_{02} &  \label{nabsqrm11} \\
&
[\Delta\tau-(\gamma-\bar{\gamma})\tau](-3\Psi_{2}+2\Phi_{11})+[\Delta\bar{\tau}-
(\bar{\gamma}-\gamma)\bar{\tau}]\Phi_{02} \, , &  \label{nabsqrm12}
\end{eqnarray}

\noindent  along with their complex conjugates. In addition, there
are boost weight 0 components in $\nabla^{(2)}(Riem)$ that are not
given here.  Since $\Psi_{2}\neq 0$ and recalling that the Riemann
type D aligned frame is completely fixed, then the vanishing of
boost weight +1 components implies (\ref{dpi}) and vanishing of
boost weight -1 components implies (\ref{deltau}):
\begin{eqnarray}
D\pi & = & 0 \label{dpi} \\
\Delta\tau-(\gamma-\bar{\gamma})\tau & = & 0 \, . \label{deltau}
\end{eqnarray}

\noindent Therefore, if $Riem$ and $\nabla(Riem)$ are of type D
and $\Psi_{2}\neq 0$ then, in general, $\nabla^{(2)}(Riem)$ will
be of type (I,I) (i.e., has only non-vanishing boost weight +1, 0,
-1 components). $\nabla^{(2)}(Riem)$ is type (II,I) if (\ref{dpi})
is satisfied and type D if (\ref{dpi}) and (\ref{deltau}) are
satisfied.  In fact, these conditions also characterize the same
algebraic types of $\nabla^{(2)}(Riem)$ for Case 2.1.  This
follows from the definition of Case 2.1; since it has the same
conditions as Case 1 except with the additional constraints
$\Psi_{2}=0$, $4\Phi_{11}^{\ \ 2}-\Phi_{02}\Phi_{20}\neq 0$ we
obtain precisely the same components of $\nabla^{(2)}(Riem)$.  If
$\Psi_{2}=0$, $4\Phi_{11}^{\ \ 2}-\Phi_{02}\Phi_{20}\neq 0$ then
the vanishing of boost weight +1 components, (\ref{nabsqr11}) and
its complex conjugate, gives again (\ref{dpi}).  In addition, the
vanishing of boost weight -1 components, (\ref{nabsqrm11}) and its
complex conjugate, again results in (\ref{deltau}).  We can now
characterize when Riemann and all of its covariant derivatives
will be type D in Cases 1 and 2.1,

\begin{thm}
In the Kundt class, if $Riem$ and $\nabla(Riem)$ are type D and $\Psi_{2}\neq
0$ or $\Psi_{2}=0$, $4\Phi_{11}^{\ \ 2}-\Phi_{02}\Phi_{20}\neq 0$,  then $D\pi=0$ and $\Delta\tau-(\gamma-\bar{\gamma})\tau=0$ if and only if $\nabla^{(n)}(Riem)$ is type D for all $n\geq 2$.
\end{thm}
\begin{proof}
In the above, we showed that (\ref{dpi}) and (\ref{deltau}) are
necessary and sufficient conditions for $\nabla^{(2)}(Riem)$ to be
type D.  Suppose $\nabla^{(n-1)}(Riem)$ and $\nabla^{(n)}(Riem)$
are type D for a fixed $n\geq 2$; then there exists an NP tetrad
in which the only non-vanishing components are boost weight 0.
Symbolically, we shall write a representative term of
$(\nabla^{n}R)_{0}$ as,
\begin{equation}
\nabla^{(n)}(Riem)=R^{n}_{0}S \, , \label{rn0s}
\end{equation}
\noindent where $R^{n}_{0}$ is a NP tetrad component of boost
weight 0 of the $n^{th}$ covariant derivative of Riemann.  In
general, $S=S(p,q,r,s)$, is a rank $n+4=p+q+r+s$ tensor
representing the outer product of $n+4$ tetrad vectors with $p$,
$q$, $r$, $s$ counting the number of $\ell$, ${\bf n}$, ${\bf m}$
and ${\bf \bar{m}}$ vectors, respectively.  Hence, $S(p,q,r,s)$ is
associated with tetrad components of boost weight $q-p$; in
particular, (\ref{rn0s}) has $S=S(q,q,r,s)$.

Taking a covariant derivative of (\ref{rn0s}) and applying
(\ref{npcdl})--(\ref{npcdm}) gives the following non-vanishing
components of boost weight +1, 0 and -1:
\begin{equation}
\nabla^{(n+1)}(Riem)=DR^{n}_{0}\, Sn + (\nabla^{n}R)_{0} + [\Delta
R^{n}_{0}+(r-s)(\gamma-\bar{\gamma})R^{n}_{0}]S\ell \, . \label{rn1s}
\end{equation}
\noindent Therefore, $\nabla^{(n+1)}(Riem)$ is type D if the
following conditions hold for all boost weight 0 components of
$\nabla^{(n)}(Riem)$:
\begin{eqnarray}
DR^{n}_{0} & = & 0 \label{drn1} \\
\Delta R^{n}_{0}+(r-s)(\gamma-\bar{\gamma})R^{n}_{0} & = & 0 \, .
\label{delrn1}
\end{eqnarray}
\noindent To proceed, we must determine the form of the $R^{n}_{0}$ tetrad
components.  Since it is assumed that $\nabla^{(n-1)}(Riem)$ is type D
then we also have
\begin{equation}
\nabla^{(n-1)}(Riem)=R^{n-1}_{0}\tilde{S}(\tilde{q},\tilde{q},\tilde{r},\tilde{s})
\, \label{rnm1s}
\end{equation}
\noindent where $n+3=2q+r+s$.  By assumption, the covariant
derivative of (\ref{rnm1s}) gives rise only to boost weight  0
components, $R^{n}_{0}$,
\begin{eqnarray}
\lefteqn{\nabla^{(n)}(Riem) = \left[-\delta R^{n-1}_{0} +
(\tr-\ts)(\bar{\alpha}-\beta)R^{n-1}_{0}\right]\tS(\tq,\tq,\tr,\ts)
\bar{m} } \nonumber \\
& & \mbox{} + \left[-\bar{\delta} R^{n-1}_{0} -
(\tr-\ts)(\alpha-\bar{\beta})R^{n-1}_{0}\right]\tS(\tq,\tq,\tr,\ts) m
+ R^{n-1}_{0} \left\{ \tq \left[-\bar{\tau}\tS(\tq-1,\tq,\tr+1,\ts)\ell
\nonumber \right.\right.\\
& & \mbox{}\left. - \tau\tS(\tq-1,\tq,\tr,\ts+1)\ell +
\pi\tS(\tq,\tq-1,\tr+1,\ts)n + \bar{\pi}\tS(\tq,\tq-1,\tr,\ts+1)n \right]
\nonumber \\
& & \mbox{} + \tr\left[\bar{\pi}\tS(\tq+1,\tq,\tr-1,\ts)n -
\tau\tS(\tq,\tq+1,\tr-1,\ts)\ell\right]  \nonumber \\
& & \mbox{}\left. +
\ts\left[\pi\tS(\tq+1,\tq,\tr,\ts-1)n-\bar{\tau}\tS(\tq,\tq+1,\tr,\ts-1)\ell\right]
\right\} \, . \label{nabrn0}
\end{eqnarray}
\noindent Furthermore, equations (\ref{drn1}) and (\ref{delrn1})
are identically satisfied under replacement $n,r,s\mapsto
n-1,\tr,\ts$; i.e., $\nabla^{(n)}(Riem)$ is type D.

First consider equation (\ref{drn1}).  From (\ref{nabrn0}), a
component $R^{n}_{0}$ is proportional to $R^{n-1}_{0}\pi$;
therefore, $DR^{n}_{0}=D(R^{n-1}_{0}\pi)=R^{n-1}_{0}D\pi=0$ which
is satisfied since $\nabla^{(2)}(Riem)$ type D implies $D\pi=0$.
Another boost weight 0 component is proportional to
$R^{n-1}_{0}\tau$, then $D(R^{n-1}_{0}\tau)=R^{n-1}_{0}D\tau=0$
holds as a consequence of an NP equation.  Similar conclusions
hold for $R^{n-1}_{0}\bar{\pi}$ and $R^{n-1}_{0}\bar{\tau}$. Next,
using the NP equations $D\alpha=D\beta=0$ and commutator relation
we find that
\begin{equation}
D\left[-\delta R^{n-1}_{0} +
(\tr-\ts)(\bar{\alpha}-\beta)R^{n-1}_{0}\right]=-\delta
DR^{n-1}_{0}+(\bar{\alpha} + \beta - \bar{\pi})DR^{n-1}_{0}=0 \, .
\end{equation}
\noindent The same conclusion holds for the second component of
(\ref{nabrn0}).  It now follows that $D\pi=0$ is the necessary and
sufficient condition for $\nabla^{(n+1)}(Riem)$ to have vanishing
boost weight +1 components.

Last, consider equation (\ref{delrn1}).  Using the NP equation
$\Delta\pi=-(\gamma-\bar{\gamma})\pi$, the component
$R^{n-1}_{0}\pi$ gives
\begin{equation}
\Delta(R^{n-1}_{0}\pi)+(r-s)(\gamma-\bar{\gamma})R^{n-1}_{0}\pi=(r-s-\tr+\ts-1)(\gamma-\bar{\gamma})R^{n-1}_{0}
\pi =0
\end{equation}
\noindent where the last equality follows from the fact that
$R^{n-1}_{0}\pi$ occurs in (\ref{nabrn0}) with either $r=\tr +1$,
$s=\ts$ or $r=\tr$, $s=\ts -1$.  For the boost weight 0 component
$R^{n-1}_{0}\tau$, equation (\ref{delrn1}) becomes
\begin{equation}
\left[\Delta\tau +
(r-s-\tr+\ts)(\gamma-\bar{\gamma})\tau\right]R^{n-1}_{0}=0 \, .
\label{delrnm1tau}
\end{equation}
\noindent Since (\ref{nabrn0}) implies that $R^{n-1}_{0}\tau$
occurs with either $r=\tr$, $s=\ts+1$ or $r=\tr-1$, $s=\ts$, then
(\ref{delrnm1tau}) reduces to
$\Delta\tau-(\gamma-\bar{\gamma})\tau=0$ which is satisfied as a
consequence of having $\nabla^{(2)}(Riem)$ type D.  Again, similar
arguments apply to the components $R^{n-1}_{0}\bar{\pi}$ and
$R^{n-1}_{0}\bar{\tau}$.  Next, we consider the $R^{n}_{0}$
component of (\ref{nabrn0}) proportional to $-\delta R^{n-1}_{0} +
(\tr-\ts)(\bar{\alpha}-\beta)R^{n-1}_{0}$ and substitute it into
(\ref{delrn1}).  By applying the commutator relation
\begin{equation}
\Delta\delta(R^{n-1}_{0})=\delta\Delta(R^{n-1}_{0})-(\tau-\bar{\alpha}-\beta)\Delta
R^{n-1}_{0} +(\gamma-\bar{\gamma})\delta R^{n-1}_{0} \, ,
\end{equation}
\noindent the identity $\Delta
R^{n-1}_{0}=-(\tr-\ts)(\gamma-\bar{\gamma})R^{n-1}_{0}$, and the NP
equations
\begin{eqnarray}
\Delta\alpha & = & \bar{\delta}\gamma + \bar{\gamma}\alpha +
(\bar{\beta}-\bar{\tau})\gamma \\
\Delta\beta & = & \delta\gamma + (\bar{\alpha}+\beta-\tau)\gamma + (\gamma
- \bar{\gamma})\beta \, ,
\end{eqnarray}
\noindent we simplify to get
\begin{equation}
\left(\tr-\ts-r+s-1\right)(\gamma-\bar{\gamma})\left[\delta
R^{n-1}_{0}-(\tr-\ts)(\bar{\alpha}-\beta)R^{n-1}_{0} \right]=0 \, .
\label{Ddeltarn0}
\end{equation}
\noindent Equality in (\ref{Ddeltarn0}) follows from
(\ref{nabrn0}), which implies $r=\tr$, $s=\ts+1$.  The second
component of (\ref{nabrn0}) also identically satisfies
(\ref{delrn1}), except here $r=\tr+1$, $s=\ts$. Thus,
$\Delta\tau-(\gamma-\bar{\gamma})\tau=0$ is the necessary and
sufficient condition for $\nabla^{(n+1)}(Riem)$ to have vanishing
boost weight -1 components.

Therefore, we have shown that (\ref{dpi}) and (\ref{deltau}) are
the necessary and sufficient conditions for $\nabla^{(n+1)}(Riem)$
to have vanishing boost weight +1 and -1 components for any $n
\geq 2$, and hence to be of type D.
\end{proof}

Now consider Case 2.2, defined by $\Psi_{2}=0$, $4\Phi_{11}^{\ \ 2}-\Phi_{02}\Phi_{20}=0$ with $\nu$, $\lambda$ and $\mu$ satisfying (\ref{nabrm22}) and (\ref{nabrm13}).  In addition to (\ref{nabrm11}) and (\ref{nabrm14}) holding, the Bianchi equations also give
\begin{equation}
D\Phi_{11}=D\Phi_{02}=DR=0 \, . \label{Dsczero}
\end{equation}
\noindent Clearly, $\Phi_{11}=0$ if and only if $\Phi_{02}=0$, in
which case $R$ is the only non-vanishing curvature scalar.
However, the Bianchi equations consequently give $\delta R=0$ so
that (\ref{nabrm11}) and (\ref{Dsczero}) imply that $R$ is a
constant. Thus we obtain a space of constant curvature having
$Riem$ type D and $\nabla(Riem)$ type O since it vanishes (a
symmetric space). In the remainder of this section we shall assume
$\Phi_{11}\neq 0$.

Using the NP equations, Bianchi equations and commutation
relations to simplify the components of $\nabla^{(2)}(Riem)$ we
find that, up to a constant factor, all boost weight +1 components
reduce to $2D\pi\Phi_{11}+D\bar{\pi}\Phi_{20}$ or its complex
conjugate.  Taking the conjugate of (\ref{bianb02}) and
subtracting (\ref{bianb01}) gives the identity
\begin{equation}
2(\tau+\bar{\pi})\Phi_{11}+(\bar{\tau}+\pi)\Phi_{02}=0 \, ,  \label{pitauident}
\end{equation}
\noindent and applying $D$ to (\ref{pitauident}) and using
(\ref{Dsczero}) and the NP equation $D\tau=0$ gives
\begin{equation}
2D\bar{\pi}\Phi_{11}+D\pi\Phi_{02}=0 \, ;    \label{Dpiident}
\end{equation}
\noindent therefore, all boost weight +1 components vanish.  All
other positive boost weightcomponents of $\nabla^{(2)}(Riem)$
vanish.  We note that similar relations to (\ref{Dpiident}) can be
obtained by taking $D$ of (\ref{nabrm22}) and (\ref{nabrm13}).

Since $\nabla(Riem)$ is type D, all boost weight -3 components of
$\nabla^{(2)}(Riem)$ must vanish.  Evaluating these components we
find they can all be made to vanish by applying the following
identities.  Taking $\Delta$ of (\ref{nabrm22}) gives
\begin{equation}
2\Delta\bar{\nu}\Phi_{11}+\left[\Delta\nu+2(\gamma-\bar{\gamma})\nu\right]\Phi_{02}=0
\, ;  \label{delnuident}
\end{equation}
\noindent  similar relations to (\ref{delnuident}) can be obtained by taking $\Delta$ of (\ref{nabrm13}).  Assuming $\nu \neq 0$, then from (\ref{nabrm22}), its conjugate and the second equation of (\ref{nabrm13}) we obtain
\begin{eqnarray}
\nu^{2}\Phi_{02}-\bar{\nu}^2\Phi_{20}=0, & \bar{\lambda}\nu-\mu\bar{\nu}=0 \, .
\end{eqnarray}
If $\nu=0$, then all boost weight -3 components vanish identically
in $\nabla^{(2)}(Riem)$.

Next, we list some of the identities that are useful in the
simplification of the remaining non-vanishing components of
$\nabla^{(2)}(Riem)$ at boost weight -2, -1 and 0:
\begin{eqnarray}
\mu\lambda\Phi_{02}-\bar{\mu}\bar{\lambda}\Phi_{20}=0, & \mu\bar{\mu}-\lambda\bar{\lambda}=0 \\
\lambda^{2}\Phi_{02}-\bar{\mu}^{2}\Phi_{20}=0, & \nu\lambda\Phi_{02}-\bar{\nu}\bar{\mu}\Phi_{20}=0 \\
\lambda(\bar{\tau}+\pi)\Phi_{02}-\bar{\mu}(\tau+\bar{\pi})\Phi_{20}=0, & \bar{\lambda}(\bar{\tau}+\pi)-\mu(\tau + \bar{\pi})=0 \label{lmtpident} \\
\nu(\bar{\tau}+\pi)\Phi_{02}-\bar{\nu}(\tau+\bar{\pi})\Phi_{20}=0, & \nu(\tau+\bar{\pi})-\bar{\nu}(\bar{\tau}+\pi)=0 \label{ntpident} \\
(\bar{\tau}+\pi)^{2}\Phi_{02}-(\tau+\bar{\pi})^2\Phi_{20}=0,
\end{eqnarray}

\begin{eqnarray}
\Delta^{2}\Phi_{02}-2\left[\Delta\gamma - \Delta\bar{\gamma}+ 2(\gamma-\bar{\gamma})^2\right]\Phi_{02}=0 \\
2\bar{\delta}\bar{\nu}\Phi_{11}+\bar{\delta}\nu\Phi_{02}-\frac{1}{12}\nu\delta R-\frac{1}{12}\bar{\nu}\bar{\delta}R + 2(\nu\bar{\pi}-\bar{\nu}\pi)\Phi_{11} + 2\nu(\alpha-\bar{\beta})\Phi_{02} =0 \\
2\delta\bar{\nu}\Phi_{11}+\delta\nu\Phi_{02}+\nu\delta\Phi_{02}-\frac{1}{12}\bar{\nu}\delta R + 2(\nu\bar{\pi}-\bar{\nu}\pi)\Phi_{02}=0 \\
(\bar{\tau}+\pi)\delta R + (\tau+\bar{\pi})\bar{\delta}R = 0 \, . \label{tpRident}
\end{eqnarray}
\noindent Other identities also follow by taking $D$, $\Delta$,
$\delta$ or $\bar{\delta}$ of known identities and applying the NP
equations, Bianchi equations and commutation relations.  We find
the boost weight -1 and -2 components of $\nabla^{(2)}(Riem)$
reduce to

\noindent $(\nabla^{2} R)_{-1}$:
\begin{equation}
\mu\delta R + \bar{\lambda}\bar{\delta}R  \label{22nabsqrm1}
\end{equation}
\noindent $(\nabla^{2} R)_{-2}$:
\begin{equation}
\nu\delta R + \bar{\nu}\bar{\delta}R \, .   \label{22nabsqrm2}
\end{equation}
\noindent As in the previous cases, the boost weight 0 components,
$(\nabla^{2} R)_{0}$,  are, in general, non-zero (but we do not
give them here).   Therefore, if $Riem$ and $\nabla(Riem)$ are
type D and $\Psi_{2}=0$, $4\Phi_{11}^{\ \ 2}-\Phi_{02}\Phi_{20}=0$
then, in general, $\nabla^{(2)}(Riem)$ is type (II,H) (i.e., has
non-vanishing components of boost weight 0, -1 and -2).
$\nabla^{(2)}(Riem)$ is type (II,I) if (\ref{22nabsqrm2}) vanishes
and type D if (\ref{22nabsqrm1}) and (\ref{22nabsqrm2}) vanish.
Notice that the type D conditions are the requirements for
(\ref{tpRident}) and the second equations of (\ref{lmtpident}) and
(\ref{ntpident}) to admit a nontrivial solution; $\tau+\bar{\pi}
\neq 0$.

\subsubsection{Final comments}

Suppose we impose the requirement that $\nabla^{(2)}(Riem)$ is
type D, so that (\ref{22nabsqrm1}) and (\ref{22nabsqrm2}) vanish.
Calculating $\nabla^{(3)}(Riem)$, we expect all positive boost
weight components to vanish, which would follow from
(\ref{Dsczero}), the NP equations $D\tau=D\alpha=D\beta=0$, and a
higher order identity on $D\pi$ similar to equation
(\ref{Dpiident}). It is likely that these conditions will result
in the vanishing of all positive boost weight components of
$\nabla^{(n)}(Riem)$ for all $n\geq 3$, thus leading in general to
primary alignment type II.  In addition to boost weight  0
components, $\nabla^{(3)}(Riem)$ may contain components of boost
weight -1 and -2.  One possibility is that the boost weight -1, -2
components vanish identically, as a result of (\ref{22nabsqrm1})
and (\ref{22nabsqrm2}) vanishing.  In this case,
$\nabla^{(3)}(Riem)$ is type D and one might expect that all
higher covariant derivatives of the Riemann tensor are also type
D, so that the vanishing of (\ref{22nabsqrm1}) and
(\ref{22nabsqrm2}) provide necessary and sufficient conditions for
the Riemann tensor and all of its covariant derivatives to be type
D within Case 2.2. Another possibility is that components of
$\nabla^{(3)}(Riem)$ of boost weight -1, -2 do not vanish.
Requiring boost weight -1 and -2 to vanish would provide
additional constraints on $\mu$, $\lambda$ and $\nu$,
respectively. By requiring $\nabla^{(n)}(Riem)$ be type D for all
$n \geq 3$ we would obtain a sequence of constraints on $\mu$,
$\lambda$ and $\nu$, and it is plausible that this will reduce to
the trivial solution $\mu=\lambda=\nu=0$.  It should be noted that
this is not in general a specialization of Case 1 or 2.1, since
throughout 2.2 we require $\Psi_{2}=0$ and $4\Phi_{11}^{\ \
2}-\Phi_{02}\Phi_{20}=0$.

\section{Kundt spacetimes and their invariant classification}

It is of interest to provide invariant conditions to distinguish
the degenerate Kundt class from the remaining Kundt metrics.  We
define the following invariant and covariant quantities involving
the second and third Lie derivative with respect to ${\bl}$:
\begin{equation}
\begin{array}{ll}
I_{0}=R^{abcd}R_{a\ c\ }^{\ e\
f}\mathcal{L}_{\ell}\mathcal{L}_{\ell}g_{bd}\mathcal{L}_{\ell}
\mathcal{L}_{\ell}g_{ef}\, , &
K_{ab}=\mathcal{L}_{\ell}\mathcal{L}_{\ell}\mathcal{L}_{\ell}g_{ab}\,
.  \label{invarconds}
\end{array}
\end{equation}
\noindent From the Kundt metric we find that $I_0$ is proportional
to \mbox{$P^{4}(W_{1,vv}^{\quad\ 2}+W_{2,vv}^{\quad\ 2})$} and
$K_{ab}$ has components proportional to the third derivative of
$H$, $W_{1}$ and $W_{2}$ with respect to $v$, thus establishing
the following result: {\footnote{We recall that the $n^{th}$ order
degenerate Kundt class, $K_{n}$, is defined in Definition 4.1.}}

\begin{prop}
Within the Kundt class
\begin{enumerate}
\item[i)] $I_{0}=0$ if and only if $K_{0}$ is satisfied.
\item[ii)] $I_{0}=K_{ab}=0$ if and only if $K_{1}$ is satisfied.
\end{enumerate}
\end{prop}
\noindent The covariant condition in \textit{ii)} can be made
invariant by replacing it with the vanishing of the trace of
$K_{ab}$.

\subsubsection{Equivalence Problem for degenerate Kundt spacetimes}

The degenerate Kundt spacetimes are not completely characterized
by their scalar polynomial curvature invariants \cite{inv}.
However, they are completely characterized by their algebraic
properties (as we have discussed above in detail). Moreover, the
degenerate Kundt spacetimes are, of course, uniquely characterized
by their Cartan invariants.

This work is of importance to the equivalence problem of
characterizing Lorentzian spacetimes (in terms of their Cartan
invariants) \cite{kramer}. By knowing which spacetimes can be
characterized by their scalar curvature invariants alone, the
computations of the invariants (i.e., simple  scalar invariants)
is much more straightforward and can be done algorithmically
(i.e., the full complexity of the equivalence method is not
necessary). On the other hand, the Cartan equivalence method also
contains, at least in principle, the conditions under which the
classification is complete (although in practice carrying out the
classification for the more general spacetimes is difficult, if
not impossible). Therefore, in a sense, the full machinery of the
Cartan equivalence method is only necessary for the classification
of the degenerate Kundt spacetimes.

The first step is to completely fix the frame in all algebraic
classes (which is what was actually done in the computations
above). In such a fixed frame all of the remaining components are,
in fact, Cartan scalars. This is the easy part; the complete
characterization depends on all of the different branches that
occur. We hope to return to this problem in the future.

Even though the degenerate Kundt metrics are not determined by
their scalar polynomial curvature invariants, a special result is
possible. Suppose there exists a frame in which all of the
positive boost weight terms of the Riemann tensor and all of its
covariant derivatives  $\nabla^{(k)} (Riem)$ are zero (in this
frame). It follows from \cite{inv} that in 4D the resulting
spacetime is degenerate Kundt. It is of interest to prove that
such a spacetime is degenerate Kundt  in arbitrary dimensions
(i.e., the appropriate Ricci rotation coefficients $L_{ij}$ are
zero). We shall prove this in the final section.

In \cite{CSI4} it was shown that a 4D type D$^k$ CSI spacetime, in
which the Riemann tensor and $\nabla^k(Riemann)$ are all
simultaneously of type D, is locally homogeneous. Thus a CSI
spacetime in which the Riemann tensor and all of its covariant
derivatives are aligned and of algebraic type D, even though they
are not $\mathcal{I}$-non-degenerate, are in some sense
`characterized' by their constant curvature invariants, at least
within the class of type D$^k$ CSI spacetimes. In general, there
are many degenerate Kundt CSI metrics (that are not type D$^k$)
with the same set of constant invariants, $\texttt{I}$. In this
case there is at least one $\nabla^k(Riemann)$ which is proper
type II and thus has negative boost weight terms; this Kundt CSI
metric will have precisely the same scalar curvature invariants as
the corresponding type D$^k$ CSI metric (which has no negative
boost weight terms). Therefore, there a distinguished or a
`preferred' metric with the same set of constant invariants,
$\texttt{I}$; namely, the corresponding type D$^k$ locally
homogeneous CSI metric, which is distinguished within the class of
algebraic type D$^k$ CSI spacetimes. Similar properties are likely
true for all D$^k$ spacetimes.

\subsection{Kundt spacetimes and scalar invariants}

Let us consider to what extent the class of  {\it degenerate}
Kundt spacetimes can be characterized by their scalar polynomial
curvature invariants. Clearly such spacetimes are algebraically
special and of Riemann type II. In particular, the Weyl tensor is
of type II (or more special) and hence $27J^2= I^3$ (see
(\ref{ijf})). If $I=J=0$, then the spacetime is of Weyl type III,
N or O. If the spacetime is of Weyl type N, then
$\mathcal{I}_1=\mathcal{I}_2=0$ if and only if
$\kappa=\rho=\sigma=0$ from the results in \cite{pravda} (the
definitions of the invariants $\mathcal{I}_1$ and $ \mathcal{I}_2$
are given therein). Similar results follow for Weyl type III
spacetimes (in terms of invariants $\tilde{\mathcal{I}}_1$ and
$\tilde{\mathcal{I}}_2$) and in the conformally flat (but
non-vacuum) case (in terms of similar invariants $\mathcal{I}_1$
and $\mathcal{I}_2$ constructed from the Ricci tensor \cite{pravda}; also see eqns.
(\ref{rictypeii})). These conditions, the list of conditions on the
scalar invariants of the Weyl tensor and its covariant derivatives
summarized in \cite{inv} (and, indeed, all of the conditions
discussed below), are {\em necessary} conditions in order for a
spacetime not to be $\mathcal{I}$-non-degenerate (i.e, if any of
these necessary conditions are not satisfied, the spacetime cannot
be degenerate Kundt). For example, if $27J^2\neq I^3$, then the
spacetime is of Petrov type I, and the spacetime is
$\mathcal{I}$-non-degenerate.

In the case that $27J^2= I^3\neq 0$ (Weyl types II or D), in
\cite{inv} two higher order invariants were given as
\emph{necessary} conditions for $\mathcal{I}$-non-degeneracy (if
$27J^2= I^3$, but $S_1\neq 0$ or $S_2\neq 0$, then the spacetime
is $\mathcal{I}$-non-degenerate). Indeed, if $27J^2= I^3\neq 0$
(Weyl types II and D), essentially if $\kappa=\rho=\sigma \neq 0$,
we can construct positive  boost weight terms in the derivatives
of the curvature and determine an appropriate set of scalar
curvature invariants. For example, consider the positive boost
weight terms of the first covariant derivative of the Riemann
tensor, $\nabla(Riem)$. If the spacetime is
$\mathcal{I}$-non-degenerate, then each component of
$\nabla(Riem)$ is related to a scalar curvature invariant. In this
case, in principle  we can solve (for the positive  boost weight
components of $\nabla(Riem)$) to uniquely determine
$\kappa,\rho,\sigma$ in terms of scalar invariants, and we can
therefore find necessary conditions for the spacetime to be
degenerate Kundt (there are two cases to consider, corresponding
to whether $\Psi_2 + \frac{2}{3} \Phi_{11}$ is zero or non-zero).
We note that even if the invariants exist in principle, it may not
be possible to construct them in practice.

Further necessary conditions can be obtained from the fact that
the Ricci tensor is of type II or D (or more special) (see
equations (\ref{rictypeii})/(\ref{rictypeii2}) and (\ref{riccidef})).

It is useful to express  the
conditions (\ref{ijf}) in non-NP form. The syzygy
$I^3-27J^2=0$  is complex, whose real and
imaginary parts can be expressed using invariants of Weyl not containing
duals. The real part is equivalent to:
\begin{equation}
-11 W_{2}^3 + 33 W_2 W_4 - 18 W_6 = 0, \label{weyl1}
\end{equation}
and the imaginary part is equivalent to:
\begin{equation}
(W_{2}^2 - 2 W_4)(W_{2}^2 + W_4)^2 + 18 W_3^2(6 W_6 - 2 W_{3}^2 -
9 W_{2} W_4 + 3 W_{2}^3) = 0, \label{weyl2}
\end{equation}
where
\begin{eqnarray}
W_2 &=& \frac{1}{8}C_{abcd}C^{abcd},\\ \nonumber
W_3  &=& \frac{1}{16}C_{abcd}C^{cd}_{~~ pq}C^{pqab},\\ \nonumber
W_4 &=& \frac{1}{32}C_{abcd}C^{cd}_{~~pq}C^{pq}_{~~r s}C^{rsab}, \\ \nonumber
W_6 &=& \frac{1}{128}C_{abcd}C^{cd}_{~~pq}C^{pq}_{~~r s}C^{rs}_{~~tu}C^{tu}_{~~vw}C^{vwab}.
\label{weyldef}
\end{eqnarray}

\subsubsection{A set of necessary conditions for degenerate Kundt}

For a degenerate Kundt spacetime there exists a frame in which
the Riemann tensor and all of its covariant derivatives have no
positive  boost weight components (i.e., all of $Riem, \nabla(Riem),...
\nabla^k(Riem),...$ are of `type II') (in addition, for example, to being
Kundt). This means that every tensor, $T$, constructed from the
Riemann tensor and its covariant derivatives by contractions,
additions and products, are of type II (i.e., have no positive boost weight
components). We can generate a set of necessary conditions from
these type II conditions.

(i) IIs: all symmetric trace-free (0,2) tensors, $S=T^{s}$, give
rise to necessary conditions of the form (\ref{rictypeii}), which
we shall denote by $\texttt{I}^s_{~i}$, where
\begin{equation}
s_{2}^2(4s_{1}^3-6s_{1}s_{3}+s_{2}^2)-s_{3}^2(3s_{1}^2-4s_{3})=0,
\label{rictypeii2}
\end{equation}
for each $S_i$, and the $s_{\alpha}$ ($\alpha = 1,2,3$) are defined by
\begin{eqnarray}
s_{1} &=& \frac{1}{12}S_{a}^{~b}S_{b}^{~a},\\ \nonumber
s_{2} &=& \frac{1}{24}S_{a}^{~b}S_{b}^{~c}S_{c}^{~a},\\ \nonumber
s_{3} &=& \frac{1}{48} [S_{a}^{~b}S_{b}^{~c}S_{c}^{~d}S_{d}^{~a} -
\frac{1}{4}(S_{a}^{~b}S_{b}^{~a})^2].\label{riccidef}
\end{eqnarray}

For example, in
the case of ($S_{ab}$ given by) the trace-free Ricci tensor
$R_{ab}$, $s_{\alpha} \equiv r_{\alpha}$ and $\texttt{I}^s_{~1}$ is given by equation
(\ref{rictypeii}) in terms of the CZ Ricci invariants $r_{\alpha}$. Other examples include the symmetric
trace-free parts of $R_{;ab}, R_{cd;a}R^{cd}_{~~;b}, ...$.

We note that not all such invariants are independent due to the
symmetries of the curvature tensor and the Bianchi identities. For
example, the necessary condition obtained from the trace-free part
of $\Box(R_{ab})$ would be equivalent to derivatives of the
condition obtained from the trace-free part of the Ricci tensor
itself.

(ii) IIw: all completely trace-free (0,4) tensors with the same
symmetries as the Riemann tensor, $W=T^{w}$, give rise to
necessary conditions of the form (\ref{ijf}), denoted by
$\texttt{I}^w_{~i}$. For example, for the case of ($W_{abcd}$
given by) the Weyl tensor, $\texttt{I}^w_{~1,2}$
are defined by equations (\ref{weyl1})/(\ref{weyl2})
(the real and imaginary parts of (\ref{ijf})). Other examples can be
constructed from the  covariant derivatives of the Weyl tensor.

(iii) IIo: any other invariants, $\texttt{I}^o_i$, that can be
constructed in a similar way.{\footnote {For example, if a
spacetime is Riemann type II, then not only do the Weyl type II
and Ricci type II syzygies hold, but there are additional
alignment conditions; e.g., $C_{abcd}R^{bd},
C_{abcd}R^{be}R_{e}^{~d}$ are of type II. }} In particular, the
invariants $\texttt{I}^s_{~1}$ and $\texttt{I}^w_{~1,2}$ arise
from the properties of the eigenvalues in the degenerate cases of
the eigenvalue (eigenbivalue) problems associated with the Ricci
and the Weyl tensors (defined as curvature operators acting on
tangent vectors and bivectors, respectively \cite{inv}). Any
curvature operator, constructed from the Riemann tensor and its
covariant derivatives, will give rise to an eigenvalue problem
which will give rise to invariants in a similar way.

In this way we generate a set of invariants:
$$\texttt{I} = \{\texttt{I}^s_1, ...\texttt{I}^s_i,...
\texttt{I}^w_1,\texttt{I}^w_2, ...\texttt{I}^w_i,... \texttt{I}^o_1,
...\texttt{I}^o_i,... \},$$
which, in turn, generates a set of necessary conditions (each
member of the set must vanish) for a degenerate Kundt spacetime.
Indeed, if any $\texttt{I}_i \neq 0$, the spacetime cannot be
degenerate Kundt.

\subsubsection{Comments}

\noindent There are a number of questions that arise: 1. Can we
find an independent subset of $\texttt{I}$? 2. Can we find a
minimal subset of $\texttt{I}$? 3. Does there exist a finite
independent and minimal subset of $\texttt{I}$? 4. Does there
exist such a subset  of $\texttt{I}$ that is complete in the sense
that if all invariants are zero, the spacetime is degenerate Kundt
(i.e., the conditions are sufficient as well as necessary) for some appropriate subclass?


This last possibility is not guaranteed, since it is not
immediately clear that the conditions of alignment and Kundt (in
the definition of degenerate Kundt) are contained in this list.
And the converse is false, in general, since there exist
counterexamples that were presented in  \cite{inv}. In addition,
consider, for example, a symmetric space in which
$\nabla{Riem}=0$. This means that we can construct no invariants
using the covariant derivatives of the Riemann tensor; i.e., we
have $\texttt{I}^s_{~1}$ (\ref{rictypeii}) constructed from the
Ricci tensor, $\texttt{I}^w_{~1}, \texttt{I}^w_{~2}$ (\ref{weyl1}
-- \ref{weyl2}) constructed from the Weyl tensor, and possibly an
invariant $\texttt{I}^o_{~1}$ constructed from the zeroth order
mixed invariants. Therefore, the spacetime is very restricted.
Clearly, there is a relationship between such spacetimes and
symmetric and $k$-symmetric spacetimes. For example, in
\cite{Senovilla} the set of 2-symmetric spacetimes were
investigated. It was found that a 2-symmetric spacetime is either
CSI, or there exists a covariantly constant null vector (CCNV) (a
similar result was conjectured for the set of $k$-symmetric
spacetimes).

Another approach in determining an invariant characterization of
the degenerate Kundt class (or Kundt class) would be to start by
considering an independent set of differential invariants of the
Riemann tensor and its covariant derivatives.  By requiring the
metric to be degenerate Kundt would result in syzygies among the
set of differential invariants; thus necessary conditions can be
derived.  In \cite{REN}, the authors have developed the
\emph{Invar} tensor package and performed a detailed study of the
differential invariants of the Riemann tensor.  By using their set
of independent differential invariants of the Riemann tensor and
expressing them in terms of NP scalars, the degeneracies that
occur within this set once the specialization
$\kappa=\sigma=\rho=0$, the Riemann tensor and $\nabla(Riem)$ are
type II are imposed, can be investigated.


\section{Higher dimensions}

The higher dimensional Kundt metrics are given by (\ref{Kundt}),
with $i,j = 2,..., n-1$, in the kinematic frame (\ref{null
frame}), (\ref{Rcoefs}). {\footnote{In a very recent preprint
\cite{POD}, building upon the earlier work of \cite{PODREF}, the
class of higher dimensional Kundt spacetimes has been analysed,
with an emphasis on studying exact solutions of the higher
dimensional Einstein-Maxwell equations for a vacuum or aligned
Maxwell field with and without a cosmological constant.}} From
eqns. (49) and (54) in \cite{CSI} it follows that if $W_{i,vv}=0$
for all $i$, then all positive  boost weight terms of the Riemann
tensor are zero in the kinematic frame, and the spacetime is of
aligned algebraically special Riemann type II. From the equations
for $\nabla(Riemann)$ in this frame (cf. eqns. (55) and  (60)-(63)
in \cite{CSI}), it then follows that if $H_{,vvv}=0$, then all of
the positive  boost weight terms of the covariant derivative of
the Riemann tensor are zero. By a direct calculation, it then
follows that all of the positive boost weight terms of
$\nabla^2(Riem)$ are zero in this aligned frame. It then follows
from an argument similar to that of Theorem 3.1 that all of the
positive  boost weight terms of all of the covariant derivatives
of the Riemann tensor are zero. Hence, if $W_{i,vv}=0$ and
$H_{,vvv}=0$, the higher dimensional Kundt metric is degenerate.

\subsection{Examples of degenerate Kundt spacetimes in higher dimensions}

Spacetimes which are  CSI, VSI or CCNV are degenerate-Kundt
spacetimes.

\subsubsection{Constant scalar curvature invariants}

Lorentzian spacetimes for which all polynomial scalar invariants
constructed from the Riemann tensor and its covariant derivatives
are constant ($CSI$ spacetimes) were studied in \cite{CSI}. If a
spacetime is {$CSI$}, the spacetime is either locally homogeneous
or belongs to the higher dimensional Kundt $CSI$ class (the
${CSI_K}$ conjecture) and can be constructed from locally
homogeneous spaces and $VSI$ spacetimes. The $CSI$ conjectures
were proven in four dimensions in \cite{CSI4}.

In \cite{CSI} it was shown that for Kundt $CSI$ metric of the form
(\ref{Kundt}) there exists (locally) a coordinate transformation
such that the transverse is independent of $u$ and a locally
homogeneous space. The remaining $CSI$ conditions then imply that
\begin{equation}
W_{i}(v,u,x^k)=v{W}_{i}^{(1)}(u,x^k)+{W}_{i}^{(0)}(u,x^k),\label{KundtCH}
\end{equation}
\begin{equation}
H(v,u,x^k)=\frac{v^2}{8}\left[4\sigma+({W}_i^{(1)})({W}^{(1)i})\right]+
v{H}^{(1)}(u,x^k)+{H}^{(0)}(u,x^k), \label{KundtCW}
\end{equation}
where $\sigma$ is a constant \cite{CSI}.
\subsubsection{Vanishing  scalar curvature invariants}
All curvature invariants of all orders vanish in an
$n$-dimensional Lorentzian spacetime if and only if there exists
an aligned non-expanding, non-twisting, shear-free geodesic null
direction $\ell^a$ along which the Riemann tensor has negative
boost order \cite{Higher}. Thus the Riemann tensor, and
consequently the Weyl and Ricci tensors, are of algebraic type
{\bf III}, {\bf N} or {\bf O} \cite{class}), and $VSI$ spacetimes
belong to the Kundt class \cite{kramer}. It follows that any
{$VSI$} metric can be written in the form  (\ref{Kundt}), where
local coordinates can be chosen so that the transverse metric is
flat; i.e., $g_{ij}=\delta_{ij}$ \cite{CSI}. The~metric functions
$H$ and $W_{i}$ in the Kundt metric (which can be obtained by
substituting $\sigma=0$ in eqns. (\ref{KundtCH}) and
(\ref{KundtCW})), satisfy the remaining vanishing scalar invariant
conditions and any relevant  Einstein field equations.

\subsubsection{Covariantly constant null vector:}

The aligned, repeated, null vector $\bl$ of (\ref{Kundt}) is a
null Killing vector (KV) if and only if $H_{v}=0$ and $W_{i,v}=0$,
and it then follows that  $\bl$ is also covariantly constant.
Therefore, the most general metric that admits a covariantly
constant null vector (CCNV) is (\ref{Kundt}) with $H = H(u,x^k)$
and $W_{i} = W_{i}(u,x^k)$ and is of Ricci and Weyl type {\bf II}
\cite{class}. In 4D the CCNV spacetimes are the well-known pp-wave
type {\bf N} VSI spacetimes.

\subsection{Discussion: supersymmetry and holonomy }

Supersymmetric solutions of supergravity theories have played an
important role in the development of string theory. The existence
of parallel (Killing) spinor fields, plays a central role in
supersymmetry. In the physically important dimensions below twelve
the maximal indecomposable Lorentzian holonomy groups admitting
parallel spinors are known \cite{Bryant}. A systematic
classification of supersymmetric solutions in $M$-theory was
provided in  \cite{jFF99}. There are two classes of solutions. If
the spacetime admits a covariantly constant time-like vector, the
spacetime is static and its classification reduces  to the
classification of $10$-dimensional Riemannian manifolds.

The second class of solutions consists of spacetimes which are not
static but which admit a covariantly constant null vector (CCNV).
The isotropy subgroup of a null spinor is contained in the
isotropy subgroup of the null vector, which in arbitrary
dimensions is isomorphic to the spin cover of $ISO(n-2)  \subset
SO(n-1,1)$. For $n\leq 5$ this means the holonomy group is
$\RR^{n-2}$, which implies that the metric is Ricci-null. This
leads to the n-dimensional Kundt spacetimes (see eqn.
(\ref{Kundt}), with $H_{,v} =0$ and $W_{i,v} = 0$, whence the
metric no longer has any $v$ dependence) \cite{CCNV,class}. This
class includes Kundt-$CSI$ and Kundt-$VSI$ spacetimes as special
cases.{\footnote {The null vector of a metric with ${\rm
Sim}(n-2)$ holonomy is a {\it recurrent null vector} ($RNV$) and
the metric belongs to the class of Kundt metrics ((\ref{Kundt})
with $W_{i} = W_{i}(u,x^k)$), are also of interest \cite{CGHP}.}}
The $VSI$ and $CSI$ spacetimes are of fundamental importance since
they are solutions of supergravity or superstring theory, when
supported by appropriate bosonic fields \cite{CFH}. Supersymmetry
in $VSI$ and $CSI$ type IIB supergravity solutions was studied in
\cite{CFH}.

The classification of holonomy groups in Lorentzian spacetimes is
quite different from the Riemannian case since the de Rham
decomposition theorem does not apply without modification. For a
Lorentzian manifold $M$ there are the following two possibilities
\cite{lBaI93}: {\it Completely reducible:} Here $M$ decomposes
into irreducible or flat Riemannian manifolds and a manifold which
is an irreducible or a flat Lorentzian manifold or
$(\mathbb{R},-dt)$. The irreducible Riemannian holonomies are
known, as well as the irreducible Lorentzian one, which has to be
the whole of $SO(1,n-1)$. {\it Not completely reducible:} This is
equivalent to the existence of a degenerate invariant subspace and
entails the existence of a holonomy invariant lightlike subspace.
The Lorentzian manifold decomposes into irreducible or flat
Riemannian manifolds and a Lorentzian manifold with
indecomposable, but non-irreducible holonomy representation; i.e.,
with (a one-dimensional) invariant lightlike subspace. These are
the $CCNV$ and $RNV$ (Kundt) spacetimes, which contain the $VSI$
and $CSI$ subclasses as special cases \cite{johan}.

Therefore, the  Kundt spacetimes that are of particular physical
interest are degenerately reducible, which leads to complicated
holonomy structure and various degenerate mathematical properties.
Such spacetimes have a number of other interesting and unusual
properties, which may lead to novel and fundamental physics.
Indeed, a complete understanding of string theory is not possible
without a comprehensive knowledge of the properties of the Kundt
spacetimes \cite{johan}. For example, in general a Lorentzian
spacetime is completely classified by its set of scalar polynomial
curvature invariants. However, this is not true for the degenerate
Kundt spacetimes \cite{CSI} (i.e., they have important geometrical
information that is not contained in the scalar invariants). This
leads to interesting problems with any physical property that
depends essentially on scalar invariants, and may lead to
ambiguities and pathologies in models of quantum gravity or string
theory.

As an illustration, in many theories of fundamental physics there
are geometric classical corrections to general relativity.
Different polynomial curvature invariants (constructed from the
Riemann tensor and its covariant derivatives) are required to
compute different loop-orders of renormalization of the
Einstein-Hilbert action. In specific quantum models such as
supergravity there are particular allowed local counterterms
\cite{sugra}. In particular, a classical solution is called {\it
universal}  if the quantum correction is a multiple of the metric.
In \cite{CGHP} metrics of holonomy $\mathrm{Sim}(n-2)$ were
investigated, and it was found that all 4-dimensional
$\mathrm{Sim}(2)$ metrics are universal and consequently can be
interpreted as metrics with vanishing quantum corrections and are
automatically solutions to the quantum theory. The $RNV$ and
$CCNV$ (Kundt) spacetimes therefore play an important role in the
quantum theory, regardless of what the exact form of this theory
might be.

\subsection{Results}

\

Many of the results in this paper (for 4D spacetimes) can be
generalized to higher dimensions. In particular, we would like to
prove that the degenerate Kundt metrics are the only metrics not
determined by their curvature invariants (i.e., not
$\mathcal{I}$-non-degenerate) in any dimension. There are higher
dimensional generalizations to Theorem \ref{thmK2} and the type D
result, which we will present in subsections 7.3.2 and 7.3.3.

\subsubsection{Partial converse result}

In the analysis in 4D it was determined for which Segre types for
the Ricci tensor  the spacetime is $\mathcal{I}$-non-degenerate
(similar results were obtained for the Weyl tensor). In each case,
it was found that the Ricci tensor, considered as a curvature
operator, admits a timelike eigendirection. Therefore, if a
spacetime is not $\mathcal{I}$-non-degenerate, its Ricci tensor
must be of a particular Segre type (corresponding to the
non-existence of a unique timelike direction). Therefore, it is
plausible that if the algebraic type of the Ricci tensor (or any
other $(0,2)$ curvature operator written in `Segre form') is not
of one the following types:

\begin{enumerate}

\item{}
$\{2111...\}$,$\{2(11)1...\}$,$\{2(111)...\}$,...,$\{2(111...)\}$,
\item{}
$\{(21)11...\}$,$\{(21)(11)1...\}$,$\{(21)(111)...\}$,...,$\{(21)(111...)\}$,
\item{}
$\{(211)11...\}$,$\{(211)(11)1...\}$,$\{(211)(111)...\}$,...,$\{(211)(111...)\}$,
\item{}
$\{3111...\}$,$\{3(11)1...\}$,$\{3(111)...\}$,...,$\{3(111...)\}$,
\end{enumerate}
and so on, then the spacetime is $\mathcal{I}$-non-degenerate.
(Similar results in terms of the Weyl tensor in bivector form are
possible). It remains to prove the converse; namely, if it is of
one of these types it must be degenerate Kundt.

\subsubsection{Kundt Theorem}

Let us show that Theorem \ref{thmK2} is valid in any dimension.
The Proof of Theorem \ref{thmK2} was computational in nature and
specific to 4 dimensions only; here we will give an alternative
proof, valid in any dimension.

\begin{thm}
In the higher-dimensional Kundt class, $K_1$ implies $K_n$ for all $n\geq 2$.
\end{thm}

\begin{proof}
First, $K_1$ implies that there exists a frame such that all
positive boost weight  components of the connection coefficients
are zero (this is presicely the Kundt condition in higher
dimensions), with respect to which all positive boost weight
components of Riemann and $\nabla(Riem)$ are zero. In order to
show that this implies $K_n$ for all $n$, we will use the two
identities:
\begin{eqnarray}
 && R_{ab(cd;e)}= 0,\quad  \text{(Bianchi identity)} \\
&& [\nabla_{a}, \nabla_{b}] T_{d_1...d_k}=\sum_{i=1}^k T_{d_1...e...d_k}R^{e}_{~d_i ab}\quad
\text{(Gen. Ricci identity)}.
\end{eqnarray}
Let us first assume only $K_0$. The covariant derivative consists
of a partial derivative and an algebraic piece; symbolically we
can write: \beq \nabla T=\partial T-\sum \Gamma *T. \eeq Since the
connection coefficients do not contain positive boost weight
components the algebraic piece, $\sum \Gamma *T$, cannot raise the
boost weight. Let us, for simplicity, denote components with
$a=0,1, i$, where $\ell^aT_{a...b}=T_{0...b}$, ${\bf
n}^aT_{a...b}=T_{1...b}$, and ${\bf m}_i^{~a}T_{a...b}=T_{i...b}$.
The components $0$, $1$ and $i$ (`downstairs') would therefore
carry a boost weight $+1$, $-1$ and $0$, respectively. The only
way that the boost weight can be raised (since the connection
coefficients are of boost order 0) is through the partial
derivatives. If we consider, for example, the +1 component:
\[ \nabla_iR_{jkn0}=\partial_i(R_{jkn0})=0, \]
since $R_{ijk0}=0$ by the $K_0$ assumption. Analogously, by
considering all +1 components of the form $\nabla_i R_{abcd}$, we
consequently find that they  are all zero. Similarly, we  cannot
obtain any +1 components by applying ${\bf n}^a\nabla_a$; hence,
positive boost weight  components of $\nabla R$ can only come from
the covariant derivative with respect to $\ell$.

Consider the positive boost weight  components of the form
\[ \ell^a\nabla_{a}T_{bc...}\equiv T_{bc...;0}.\]
The possible (independent) boost weight  +1 components of $\nabla
R$ are:
\[ R_{ij01;0}=R_{01ij;0}, ~~ R_{0i1j;0}, ~~ R_{ijkl;0}, ~~R_{0101;0}\]
For the first covariant derivative of the Riemann tensor, we can
use the Bianchi identity:
\[ R_{abcd;0}=-R_{ab0c;d}+R_{ab0d;c}.\]
Using the Bianchi identity and the above results,  we see that all
of the positive boost weight  components have to be zero, except
possibly $R_{0101;0}$ (this component corresponds to the term
$H_{,vvv}$ and is, in general, non-zero). However, by also
assuming $K_1$, $R_{0101;0}=0$ (and consequently, $H_{,vvv}=0$).
Therefore, let us assume that $(R)_{b>0}=0$ and $(\nabla
R)_{b>0}=0$, and consider $\nabla\nabla R$. By the argument above,
using $T=\nabla R$, we see that the only possible contributors to
the positive boost weight components are $\ell^a\nabla_a (\nabla
R)$. The right-hand side of the generalised Ricci identity is
purely algebraic and hence, because of the $K_0$ assumption, the
positive boost weight components of $[\nabla,\nabla]R$ must be
zero. Consider, for example, the component
\[ R_{01ij;k0}=R_{01ij;0k}+(\text{alg. piece}).\]
So, $R_{01ij;k0}=R_{01ij;0k}=\nabla_k(R_{01ij;0})=0$. For all
other components, except for $R_{abcd;00}$, we can use the same
trick to show that the positive boost weight components are also
zero. For $R_{abcd;00}$ consider, for example, the boost weight +1
component
\[ R_{01ij;00}=-(R_{010i;j})_{;0}+(R_{010j;i})_{;0} \]
(using the Bianchi identity); consequently, by the Ricci identity
this component is also zero. In the same way, all of the positive
boost weight components can be shown to be zero also. Thus,
$(\nabla\nabla R)_{b>0}=0$, and the spacetime is $K_2$.

We have now shown that $K_1$ implies $K_2$. In the same manner, we
can use these arguments recursively for any $\nabla^n R$ by using
the Ricci identity and the Bianchi identity. Therefore, we cannot
acquire positive boost weight components by taking  the covariant
derivatives of the Riemann tensor.  Hence, the spacetime must be
$K_n$ for all $n$.
\end{proof}


\subsubsection{Type D Theorem}

Suppose, in higher dimensions, there exists a frame in which all
of the positive and negative boost weight components of the
Riemann tensor and all of its covariant derivatives  $\nabla^{(k)}
(Riem)$ are zero (in this frame). Let us prove that the resulting
spacetime is Kundt (this was shown to be true in 4D in
\cite{inv}).

\begin{thm}
If for a spacetime $(\mathcal{M},{\bf g})$, the Riemann tensor and
all of its covariant derivatives $\nabla^{(k)}(Riem)$ are
simultaneously of type D (in the same frame), then the spacetime
is \emph{degenerate Kundt}.
\end{thm}

\begin{proof}
If the Riemann tensor and all of its covariant derivatives are of
type D in the same frame; i.e., there exists a frame such that
\[ R=(R)_0, \quad \nabla^{(k)}R=(\nabla^{(k)}R)_0, \]
then at every point all the curvature tensors are boost-invariant.
In particular, the curvature tensors experience a boost-isotropy.
Therefore, consider a point $p$ and assume this is
regular\footnote{In the sense of \cite{kramer}; i.e., the number
of independent Cartan invariants do not change at $p$.}. At this
point there is a one-parameter family of  boosts $B_t:T_pM\mapsto
T_pM$, which maps the frame onto another frame with identical
components of the curvature tensors. Explicitly, we can define the
boost to be over the neighborhood $U$ relative to the canonical
type $D^k$ frame:
\[
{\mbold\ell}\mapsto e^{\psi(t;x^\mu)}{\mbold\ell}, \quad {\bf
n}\mapsto e^{-\psi(t;x^\mu)}{\bf n}, \quad {\bf m}^i\mapsto {\bf
m}^i,
\]
for any $\psi(t;x^{\mu})$ such that when restricted to $p$ gives
$B_t$. Note that any such boost will leave the curvature tensors
invariant; hence, the Cartan scalars of the transformed tensor are
identical to the original ones. By the equivalence principle,
since this is an isometry of the curvature tensors at the point
$p$, there exists an isometry $\phi_t$ on a neighborhood $U$ of
$p$ such that it induces the map $B_t$, at the tangent space
$T_pM$ \cite{kramer,KN}. Furthermore, we can choose $\phi_t(p)=p$;
hence, the isometry $\phi_t$ is in the isotropy group at $p$.
Define also the map $M$ as the induced map (the push-forward) of
$\phi_t$ acting on $TM$ over $U$; i.e., in components,
$\phi_{t*}({\bf e}_{\mu})=M^\nu_{~\mu}{\bf e}_\nu$ over $U$. Note
that at $p$, $M$ coincides with $B_t$, and hence the map $M$ must
act, up to conjugation, as a boost. Since $\phi_t$ is an isometry,
this boost must also be in the stabilizer of the curvature tensor.
Therefore, align the null-frame over $U$ such that this boost acts
as $({\mbold \ell},{\bf n},{\bf m}^i)\mapsto (e^\lambda{\mbold
\ell},e^{-\lambda}{\bf n},{\bf m}^i)$ (note that the curvature
tensors must still be of type D with respect to this frame).

Since an isometry leaves the connection invariant \cite{kramer,KN}
(i.e., if ${\mbold\Omega}$ is the connection form, then
$\tilde{\phi}_t^*{\mbold\Omega}={\mbold\Omega}$, where
$\tilde{\phi}_t$ is the induced transformation on the frame
bundle), we get over $U$:
\[ \Gamma^{\mu}_{~\alpha\beta}=(M^{-1})^\mu_{~\nu}\left[M^\gamma_{~\alpha}\phi_t^*(\Gamma^{\nu}_{~\gamma\delta})+M^\gamma_{~\alpha,\delta}\right] M^\delta_{~\beta}.\]
Furthermore, since $p=\phi_t(p)$,  we  have
$\Gamma^{\mu}_{~\gamma\delta}=\phi_t^*(\Gamma^{\nu}_{~\gamma\delta})$
at $p$. Moreover, in the aforementioned type D frame, we have
$M^0_{~0,\mu}=-M^1_{~1,\mu}$, while all other components of
$M^\gamma_{~\alpha,\delta}$ are zero. Consequently, the following
components of $\Gamma^{\mu}_{~\alpha\beta}$ must vanish: \beq
\Gamma^{i}_{~00}=\Gamma^1_{~i0}=\Gamma^{i}_{~0j}=\Gamma^1_{~ij}=\Gamma^i_{~j0}&=&0
\quad (\text{boost weight} +2, +1)
\nonumber \\
\Gamma^{i}_{~11}=\Gamma^0_{~i1}=\Gamma^{i}_{~1j}=\Gamma^0_{~ij}=\Gamma^i_{~j1}&=&0
\quad (\text{boost weight} -2, -1). \eeq Note that the only
remaining components of positive/negative boost weights are
$\Gamma^0_{~00}=-\Gamma^1_{~10}$ (boost weight +1) and
$\Gamma^0_{~01}=-\Gamma^1_{~11}$ (boost weight -1) (note that
these are the components that do not transform algebraically but
rather through derivatives of the boost parameter). At least one
of these can be set to zero by a boost,  which means that there
exists a frame such that all positive boost weight components of
the connection coefficients are zero. The vanishing of these
positive boost weight components precisely corresponds to the
existence of a shear-free, expansion-free, twist-free, geodesic
null congruence, at the point $p$. However, since this is valid at
any point $p$, this means that the spacetime is Kundt; there
exists a null vector field over $U$ which is geodesic,
expansion-free, shear-free and twist-free.
\end{proof}
Actually, by the argument of the proof we can see that there are
always two such null vector-fields for these  type $D^k$
spacetimes. They are, in general:
\[
{\bf k}_1=f{\mbold \ell}, \quad {\bf k}_2=g{\bf n}.\] It is always
possible to boost so that \emph{either } $f=1$ \emph{or} $g=1$.
However, if $k_1^{\mu}k_{2\mu}=fg\neq 1$, it is not possible to
boost so that \emph{both} $f=1$ and $g=1$. Therefore, there are
two null-vector fields that are shear-free, expansion-free,
twist-free and geodesic; one is aligned with ${\mbold\ell}$, the
other is aligned with ${\bf n}$.

\newpage

\section*{Acknowledgments}
This work was supported by the Natural Sciences and Engineering
Research Council of Canada and the Killam Foundation through a
Killam Postdoctoral Fellowship (GOP).

\end{document}